\documentclass[11pt,reqno]{amsart}
\usepackage{txfonts}
\usepackage{mathrsfs}
\usepackage{amsfonts}
\allowdisplaybreaks[4]
\usepackage{graphicx} 
\usepackage{epsfig}
\usepackage{amssymb}
\usepackage{tcolorbox}
\usepackage{tikz}
\usepackage{amsmath,xypic}
\usepackage{cite,color}
\usepackage{enumerate}
\newtheorem{theorem}{Theorem}

\newtheorem{proposition}[theorem]{Proposition}
\newtheorem{corollary}[theorem]{Corollary}

\newtheorem{lemma}[theorem]{Lemma}
\newtheorem{remark}[theorem]{Remark}

\newcommand{\be}{\begin{equation}}
\newcommand{\ee}{\end{equation}}
\newcommand{\bea}{\begin{eqnarray}}
\newcommand{\eea}{\end{eqnarray}}
\newcommand{\ba}{\begin{array}}
	\newcommand{\ea}{\end{array}}
\newcommand{\bean}{\begin{eqnarray*}}
	\newcommand{\eean}{\end{eqnarray*}}

\newcommand{\pa}{\partial}

\begin{document}
\title{Lax structure and tau function  for  large BKP hierarchy }
\author{Wenchuang Guan$^1$, Shen Wang$^1$, Wenjuan Rui$^{1,2}$, Jipeng Cheng$^{1,2*}$}
\dedicatory { $^1$ School of Mathematics, China University of
Mining and Technology, Xuzhou, Jiangsu 221116, P.\ R.\ China;\\$^2$ Jiangsu Center for Applied Mathematics (CUMT), Xuzhou, Jiangsu 221116, P.\ R.\ China}
\thanks{*Corresponding author. Email: chengjp@cumt.edu.cn, chengjipeng1983@163.com.}
\begin{abstract}
  In this paper, we mainly investigate Lax structure and tau function for the large BKP hierarchy, which is also known as Toda hierarchy of B type, or Hirota--Ohta--coupled KP hierarchy, or Pfaff lattice.  Firstly, the large BKP hierarchy can be derived from fermionic BKP hierarchy by using a special bosonization, which is presented in the form of bilinear equation. Then from bilinear equation, the corresponding Lax equation is given, where in particular the relation of flow generator with Lax operator is obtained. Also starting from Lax equation, the corresponding bilinear equation and existence of tau function are discussed. After that, large BKP hierarchy is viewed as sub--hierarchy of modified Toda (mToda) hierarchy, also called two--component first modified KP hierarchy. Finally by using two basic Miura transformations from mToda to Toda, we understand two typical relations between large BKP tau function $\tau_n(\mathbf{t})$ and Toda tau function $\tau_n^{\rm Toda}(\mathbf{t},-\mathbf{t})$, that is, $\tau_n^{{\rm Toda}}(\mathbf{t},-{\mathbf{t}})=\tau_n(\mathbf{t})\tau_{n-1}(\mathbf{t})$ and $\tau_n^{{\rm Toda}}(\mathbf{t},-{\mathbf{t}})=\tau_n^2(\mathbf{t})$. Further we find $\big(\tau_n(\mathbf{t})\tau_{n-1}(\mathbf{t}),\tau_n^2(\mathbf{t})\big)$ satisfies bilinear equation of mToda hierarchy.
\end{abstract}

\maketitle

\section{Introduction}

Free fermions and bosons have been widely used in the classical integrable systems \cite{Alexandrov,Jimbo,DJKM,Kac2023,vandeleur2012,Jing2020,Kac1998,Kac2003} such as KP, modified KP, Toda, BKP and CKP hierarchies, where the corresponding tau functions are given in the forms of vacuum expectation values of some operators in terms of fermions or bosons. For one fixed fermionic or bosonic bilinear equations, different bosonizations of fermions or bosons will give rise to different kinds of integrable systems \cite{Jimbo,Kac2023,Kac1998,Kac2003}. For example, different bosonizations of fermionic KP hierarchy \cite{Alexandrov,Jimbo,DJKM,Kac1998,Kac2023,Kac2003}
$$\sum_{j\in\mathbb{Z}+1/2}\psi^{\pm}_j\tau\otimes\psi^{\mp}_{-j}\tau=0,$$
will give rise to usual KP, Toda and multi-component KP hierarchies.
Here in this paper, we will consider another bosonization of fermionic BKP hierarchy, which will lead to large BKP hierarchy.

Let us recall that fermionic BKP hierarchy \cite{DJKM,Jimbo,Kac1998,vandeLeur2022} is given by
  \begin{align}
S_B(\tau\otimes\tau)=\frac{1}{2}\tau\otimes\tau,\ \tau\in\mathcal{F},\label{fermionbilinear}
\end{align}
where $
S_B=\sum_{j\in\mathbb{Z}}(-1)^j\phi_j\otimes\phi_{-j}$, and  neutral free fermion $\phi_i$ $(i\in\mathbb{Z})$  \cite{DJKM,Jimbo,Kac1998,vandeLeur2022} satisfying
\begin{align}
\phi_i\phi_j+\phi_j\phi_i=(-1)^i\delta_{i+j,0}.\label{phicomun}
\end{align}
The Fock space $\mathcal{F}$ is spanned by
\begin{align}
\mathcal{F}={\rm span} \left \{\phi_{j_1}\cdots\phi_{j_l}|0\rangle| j_1>j_2\cdots>j_l>0 \right\},\label{BKPFockspace}
 \end{align}
with the vacuum vector $|0\rangle\in\mathcal{F}$ satisfying \cite{Kac1998}
\begin{align}
&\phi_{-j}|0\rangle=0,\quad \phi_0|0\rangle=\frac{1}{\sqrt{2}}|0\rangle,\quad j>0.\label{opreator-vacuum }
\end{align}
If we identify  $\mathcal{F}$ with $\mathcal{B}=\mathbb{C}[t_1,t_3,\cdots]$ by the isomorphism $\sigma_t:\mathcal{F}\rightarrow\mathcal{B}$ given by
\begin{align*}
\sigma_t\phi(z)\sigma_t^{-1}=
\frac{1}{\sqrt{2}}e^{\hat{\xi}(t_\circ,z)}e^{-2\hat{\xi}(\hat{\pa},z^{-1})},\quad \sigma_t(|0\rangle)=1,
\end{align*}
where $\phi(z)=\sum_{i\in\mathbb{Z}}\phi_iz^i,$\ $\hat{\xi}(t_\circ,z)=\sum_{k\geq0}t_{2k+1}z^{2k+1}$ and $\hat{\pa}=(\pa_{t_1},{\pa_{t_3}/3},\cdots)$.
Then the fermionic BKP hierarchy will become the usual BKP hierarchy also called {\bf small BKP hierarchy} \cite{DJKM,Zabrodin2021}, that is
\begin{align*}
{\rm Res}_{z}{z^{-1}}e^{\hat{\xi}(t_\circ-t_\circ',z)}\tau(t_\circ-2[z^{-1}]_\circ)
\tau(t_\circ'+2[z^{-1}]_\circ)=\tau(t_\circ)\tau(t_\circ'),
\end{align*}
where ${\rm Res}_z\sum_i a_iz^i=a_{-1}$ and $[z^{-1}]_\circ=(z^{-1},3z^{-3},\cdots)$.

Besides the bosonization above for $\mathcal{F}$, there is another
bosonization \cite{Kac1998}. To do this, we need to realize neutral free fermions  as  charged free fermions in the way below,
 \begin{eqnarray}
  \psi_{-\frac{i}{2}}^{+}=\left ( \sqrt{-1 }  \right ) ^i\phi_i,&& \psi_{\frac{i}{2}}^{+}=\left ( \sqrt{-1 }  \right ) ^{i+1}\phi_{-(i+1)},\nonumber\\
  \psi_{\frac{i}{2}}^{-}=\left ( \sqrt{-1 }  \right ) ^{i}\phi_{-i} ,&&\psi_{-\frac{i}{2}}^{-}=\left ( \sqrt{-1 }  \right ) ^{i+1}\phi_{(i+1)}, \quad i\in\mathbb{Z}^+_{\rm odd},\label{Lemma:psiphirepresent}
\end{eqnarray}
then
\begin{align*}
\psi_m^{\lambda}\psi_n^{\mu}+\psi_n^{\mu}\psi_m^{\lambda}=\delta_{\lambda,-\mu}\delta_{m,-n},\ \lambda,\mu=\pm,\ n,m\in\mathbb{Z}+{1}/{2}.
\end{align*}
Note that by \eqref{opreator-vacuum } and \eqref{Lemma:psiphirepresent},
$\psi_{n}^{\pm}|0\rangle=0$ ($n>0$).
So $\mathcal{F}$ in \eqref{BKPFockspace} can also be spanned by
\begin{align*}
\psi_{n_1}^{+}\cdots\psi_{n_r}^{+}\psi_{m_1}^{-}\cdots\psi_{m_s}^{-}|0\rangle,
 \end{align*}
 where $n_1<n_2\cdots<n_r<0$ and  $m_1<m_2\cdots<m_s<0$. Here we should note that $\sqrt{2}\phi_0|0\rangle=|0\rangle$ in \eqref{opreator-vacuum }.
Further if set
${\rm charge\ of}\  \psi_{n}^{\pm}=\pm1$,
then $\mathcal{F}$ can be decomposed according to different charges
\begin{align}
\mathcal{F}=\oplus_{k\in\mathbb{Z}}\mathcal{F}_k,\label{Fockdistict}
\end{align}
 where $\mathcal{F}_k$  is  the set of elements in $\mathcal{F}$ with charge $k$.
If define charged fermionic fields $$\psi^{\pm}(z)=\sum_{n\in\mathbb{Z}+1/2}\psi^{\pm}_nz^{-n-1/2},$$ then we can identify $\mathcal{F}$ with $$\widehat{\mathcal{B}}=\mathbb{C}[q,q^{-1},t_1,t_2,...],$$
 by the usual boson--fermion correspondence $\sigma_{\mathbf{t},q}:\mathcal{F}\rightarrow\widehat{\mathcal{B}},$ which is given by
\begin{align}
\sigma_{\mathbf{t},q}\psi^{\pm}(z)\sigma_{\mathbf{t},q}^{-1}&=q^{\pm 1}z^{\pm q\frac{\pa}{\pa q}}e^{\pm\xi(\mathbf{t},z)}e^{\mp\xi(\tilde{\pa},z^{-1})},\quad
\sigma_{\mathbf{t},q}(|0\rangle)=1,\label{bosionfermis}
\end{align}
with $\xi(\mathbf{t},z)=\sum_{k=1}^{\infty}t_kz^k$ and
$\tilde{\pa}=(\pa_{t_1},\pa_{t_2}/2,\pa_{t_3}/3,\cdots)$.

Next operator $S_B$  on $\mathcal{F}\otimes\mathcal{F}$
 can be rewritten into \cite{Kac1998}
\begin{align}\label{SBoperataor}
S_{B}=\phi_0\otimes\phi_0+S^{+}+ S^{-},\quad S^{\pm}=\sum_{n\in \mathbb{Z}+{1/2}}\psi_n^{\pm}\otimes\psi_{-n}^{\mp}.
\end{align}
Due to \eqref{Fockdistict} and \eqref{bosionfermis}, we can define $\tau_{n}(\mathbf{t})$ by  boson-fermion correspondence $\sigma_{\mathbf{t},q}$ in the way below
\begin{align*}
\sigma_{\mathbf{t},q}(\tau)=\sum_{n\in\mathbb{Z}}\tau_n(\mathbf{t})q^n,
 \end{align*}
 for $\tau\in\mathcal{F}$ satisfying the fermionic BKP hierarchy \eqref{fermionbilinear}.
So if applying  $\sigma_{\mathbf{t},q}\otimes\sigma_{\mathbf{t},q}$ to  \eqref{fermionbilinear},    one can obtain
\begin{align}\label{Hirotabkp}
&\text{{\rm Res}$_{z}$}\left(z^{n'-n''-2}\tau_{n'-1}(\mathbf{t'}-[z^{-1}])
\cdot\tau_{n''+1}(\mathbf{t}''+{[z^{-1}]})e^{\xi(\mathbf{t'}-\mathbf{t''},z)}\right.\nonumber\\
&\quad\quad\quad\left.
+z^{n''-n'-2}\tau_{n'+1}(\mathbf{t'}+{[z^{-1}]})
\cdot\tau_{n''-1}(\mathbf{t}''-[z^{-1}])e^{-\xi(\mathbf{t'}-\mathbf{t''},z)}\right)\nonumber\\
&=\frac{1}{2}\left(1-(-1)^{n'+n''}\right)\tau_{n'}
(\mathbf{t'})\tau_{n''}(\mathbf{t}''),
\end{align}
which is just the bilinear equation of {\bf the large BKP hierarchy } \cite{Jimbo,Kac1998,takasaki2009,vandeLeur2004}. Here $[z^{-1}]=(z^{-1},{z^{-2}/2},{z^{-3}/3},\cdots)$.

Note that if let $n'+n''\in2\mathbb{Z}$, \eqref{Hirotabkp} will become  the DKP hierarchy \cite{Kac1998,vandeLeur2004}, that is
   \begin{align*}
&\text{{\rm Res}$_{z}$}\left(z^{n'-n''-2}\tau_{n'-1}(\mathbf{t}'-[z^{-1}])
\cdot\tau_{n''+1}(\mathbf{t}''+{[z^{-1}]})e^{\xi(\mathbf{t}'-\mathbf{t''},z)}\right.\nonumber\\
&\quad\quad\quad\ \ \left.
+z^{n''-n'-2}\tau_{n'+1}(\mathbf{t}'+{[z^{-1}]})
\cdot\tau_{n''-1}(\mathbf{t}''-[z^{-1}])e^{-\xi(\mathbf{t}'-\mathbf{t''},z)}\right)=0,
\end{align*}
which is also known as Pfaff lattice\cite{Adler1999INRN,Adler2002math,Adler2002,Li2021,Chang2018}. While for $n'+n''\in2\mathbb{Z}+1$,  \eqref{Hirotabkp} will become the modified DKP ${\rm(mDKP)}$ hierarchiy \cite{Kac1998,vandeLeur2004}
\begin{align*}
&\text{{\rm Res}$_{z}$}\left(z^{n'-n''-2}\tau_{n'-1}(\mathbf{t}'-[z^{-1}])
\cdot\tau_{n''+1}(\mathbf{t}''+{[z^{-1}]})e^{\xi(\mathbf{t}'-\mathbf{t''},z)}\right.\nonumber\\
&\quad\quad\quad\ \ \left.
+z^{n''-n'-2}\tau_{n'+1}(\mathbf{t}'+{[z^{-1}]})
\cdot\tau_{n''-1}(\mathbf{t}'-[z^{-1}])e^{-\xi(\mathbf{t}'-\mathbf{t''},z)}\right)
=\tau_{n'}(\mathbf{t})\tau_{n''}(\mathbf{t''}).
\end{align*}
So we can believe
\begin{align*}
{\rm large\ BKP =DKP +mDKP}.
\end{align*}

For large BKP hierarchy, there are many different names.
The first appearance of large BKP should be given by Jimbo and
Miwa in \cite{Jimbo} as $D'_\infty$--integrable hierarchy. Later it was rediscovered by Hirota and Ohta in \cite{Hirota1991} as the coupled
KP hierarchy. After that,  it is also known as the Pfaff lattice \cite{Adler1999INRN,Adler2002math,Adler2002} or charged BKP hierarchy \cite{Kac1998,vandeLeur2004}. The name of large BKP should be given by Orlov,  Shiota and Takasaki \cite{takasaki2009}. Recently in \cite{Krichever2023,Zabrodin2023,Chang2018,Li2021}, large BKP hierarchy is named as  Toda hierarchy  of B type (B--Toda). The generalizations of large BKP hierarchy include multi--component generalization (e.g. \cite{Takasaki2009,Kac2023}) and  D--Toda hierarchy (see (10) in \cite{cheng2021SM}).

Though large BKP hierarchy has been extensively studied, there are still several basic questions unsolved completely. We believe the first question unsolved should be the Lax structure. In \cite{Adler2002}, the Lax equation of Pfaff lattice (i.e. DKP lattice) is expressed by semi--infinite matrices without the information of mDKP lattice. While in \cite{vandeLeur2004}, the corresponding Lax equation is given in matrix pseudo--differential operators for fixed discrete variable $n$, rather than the entire large BKP hierarchy. Recently, the expected Lax equation for large BKP hierarchy is obtained in \cite{Krichever2023}, but the relation of the flow generator $B_k$ (see \eqref{L+adjoint} in Section \ref{section3}) in Lax equation with Lax operators are not clear. So the first topic of this paper is to give the whole Lax structure from bilinear equation \eqref{Hirotabkp}, especially the relation of $B_k$ with Lax operators. Also starting from Lax structure, we obtain the corresponding bilinear equation and existence of tau functions for large BKP.

Another question is about tau functions of large BKP hierarchy. To our best knowledge, there are two different relations for tau functions of large BKP and Toda hierarchies. One is given in \cite{Adler2002math,vandeLeur2015,Li2021,Chang2018} by
\begin{align}
 \tau_n^{{\rm Toda}}(\mathbf{t},-{\mathbf{t}})=\tau_n^2(\mathbf{t}),\label{taufirst}
\end{align}
where $\tau_n^{{\rm Toda}}$ is Toda tau function. This relation is just the analogue of KP and BKP tau functions, which is just expected relation for Toda and B--Toda. While recently in \cite{Prokofev2023,Krichever2023}, another different relation is proposed, that is,
\begin{align}
 \tau_n^{{\rm Toda}}(\mathbf{t},-{\mathbf{t}})=\tau_n(\mathbf{t})\tau_{n-1}(\mathbf{t}).\label{tau2}
\end{align}
Note that relations \eqref{taufirst} and  \eqref{tau2} are quite different. So here the second topic is
to understand these two different relations. To do this, large BKP hierarchy is viewed as sub--hierarchy of modified Toda (m--Toda) hierarchy \cite{Rui2024}, which is just the 2--component 1st modified KP hierarchy \cite{Jimbo,vandeLeur2015},  relating Toda hierarchy by Miura transformations. We find that relations \eqref{taufirst} and  \eqref{tau2} are corresponding to two basic  Miura transformations from large BKP hierarchy to Toda hierarchy. Further it can be found that $\big(\tau_n(\mathbf{t})\tau_{n-1}(\mathbf{t}),\tau_n^2(\mathbf{t})\big)$ satisfies bilinear equation of mToda hierarchy.

This paper is organized in the following way.   In Section \ref{section3} for large BKP hierarchy, we firstly derive the Lax equation from bilinear equation, and then obtain the bilinear equation starting from Lax equation, and further show the existence of  tau function. Next in Section \ref{section4}, large BKP hierarchy is viewed as the sub--hierarchy of mToda hierarchy, then after the review for Miura links between mToda and Toda hierarchies, we understand the relations of tau functions between large BKP and Toda hierarchies. At last, some conclusions and discussions are provided in Section \ref{section6}.

\section{  Lax structure of  large BKP hierarchy}\label{section3}
In this section, we will discuss the Lax structure of the large BKP hierarchy. Firstly, we will derive  Lax equation from bilinear equation \eqref{Hirotabkp}, and  then conversely, obtain  bilinear equation \eqref{Hirotabkp} from Lax equation. Here we emphasize that we give the explicit expression of $B_k$ in Lax equation by Lax operators.
\subsection{ From  bilinear equation  to  Lax equation }\label{Section3.1}
If introduce  wave functions $\Psi^{\pm}(n,\mathbf{t},z)$ in the way below
\begin{align}
\Psi^{+}(n,\mathbf{t},z)
=\frac{\tau_{n-1}(\mathbf{t}-[z^{-1}])}{\tau_n(\mathbf{t})}
e^{\xi(\mathbf{t},z)}z^{n},\
\Psi^{-}(n,\mathbf{t},z)
=\frac{\tau_{n+1}(\mathbf{t}+[z^{-1}])}{\tau_n(\mathbf{t})}
e^{-\xi(\mathbf{t},z)}z^{-n-1},\label{wavefunction}
\end{align}
then the bilinear equation  \eqref{Hirotabkp} can be expressed in terms of wave functions
\begin{align}
{\rm Res}_{z}{z^{-1}}\left(\Psi^+(n',\mathbf{t'},z)
\Psi^-(n'',\mathbf{t}'',z)+\Psi^-(n',\mathbf{t'},z)
\Psi^+(n'',\mathbf{t}'',z)\right)=\frac{1}{2}\left(1-(-1)^{n'-n''}\right).\label{wavefunctionbilinear}
\end{align}
It can be found that  \eqref{wavefunctionbilinear} is equivalent to
\begin{eqnarray}
\sum_{j\in\mathbb{Z}}\text{Res$_{z}$}{z^{-1}}\left(\Psi^+(n,\mathbf{t'},z)
\Psi^-(n+j,\mathbf{t''},z)+\Psi^-(n,\mathbf{t'},z)
\Psi^+(n+j,\mathbf{t''},z)\right)\Lambda^j=\sum_{j\in \mathbb{Z}}\frac{1}{2}\left(1-(-1)^{j}\right)\Lambda^j.\label{wavefunctionshift}
\end{eqnarray}

Next  define  wave operators
\begin{eqnarray}
W^+(n,\mathbf{t},\Lambda)
=S^+(n,\mathbf{t},\Lambda)e^{\xi(\mathbf{t},\Lambda)},\quad W^-(n,\mathbf{t},\Lambda)
=e^{-\xi(\mathbf{t},\Lambda)}S^-(n,\mathbf{t},\Lambda),\label{waveoperator}
\end{eqnarray}
and
\begin{eqnarray}
S^+(n,\mathbf{t},\Lambda)=\sum_{j=0}^\infty b_j(n,\mathbf{t})\Lambda^{-j},\quad S^-(n,\mathbf{t},\Lambda)=\sum_{j=0}^\infty \Lambda^{-j}\bar{b}_j(n,\mathbf{t}),\label{Soperator}
\end{eqnarray}
such that
\begin{align}
\Psi^{+}(n,\mathbf{t},z)=W^+(n,\mathbf{t},\Lambda)(z^{n}),\quad
\Psi^{-}(n,\mathbf{t},z)
=\left(W^-(n,\mathbf{t},\Lambda)\right)^*(z^{-n-1})\label{Wavefunctions},
\end{align}
where $\Lambda$ is  the shift operator $\Lambda(f(n,\mathbf{t}))=f(n+1,\mathbf{t})$ and
$(\sum_{j\in\mathbb{Z}}f_j(n,\mathbf{t})\Lambda^{j})^{*}
=\sum_{j\in\mathbb{Z}}\Lambda^{-j}f_j(n,\mathbf{t}).
$
\begin{lemma}\label{Lemma:ABlambda}\cite{Adler1999}
	Let $A(n,\Lambda)=\sum_{i\in\mathbb{Z}}a_i(n)\Lambda^i,$ $B(n,\Lambda)=\sum_{k\in\mathbb{Z}}b_k(n)\Lambda^k$ are two  pseudo--difference operators, then
	\begin{equation}
	A(n,\Lambda)\cdot B(n,\Lambda)^*=\sum_{j\in \mathbb{Z}} \text{\rm Res$_{z}$}{z^{-1}}\left(A(n,\Lambda)(z^{\pm n})\cdot B(n+j,\Lambda)(z^{\mp n\mp j})\right)\Lambda^j,\label{ABoperatorlambda}
	\end{equation}
where $B(n,\Lambda)^*=\sum_{k\in\mathbb{Z}}\Lambda^{-k}b_k(n)$.
\end{lemma}
\noindent So by \eqref{wavefunctionshift} and Lemma \ref{Lemma:ABlambda}, it can be found that
	\begin{align}
 &W^+(n,\mathbf{t'},\Lambda)\cdot\Lambda^{-1}\cdot W^-(n,\mathbf{t'}',\Lambda)+\left(W^+(n,\mathbf{t}'',\Lambda)
 \cdot\Lambda^{-1}\cdot W^-(n,\mathbf{t'},\Lambda)\right)^*=\sum_{j\in \mathbb{Z}}\frac{1}{2}\left(1-(-1)^{j}\right)\Lambda^j.\label{hbewoperator}
	\end{align}
Further by \eqref{waveoperator}, we have
\begin{align}
&& S^+(n,\mathbf{t'},\Lambda)\Lambda^{-1} e^{\xi(\mathbf{t'}-\mathbf{t}'',\Lambda)} S^-(n,\mathbf{t}'',\Lambda)+S^-(n,\mathbf{t'},\Lambda)^*\Lambda e^{-\xi(\mathbf{t'}-\mathbf{t}'',\Lambda^{-1})} S^+(n,\mathbf{t}'',\Lambda)^*=\sum_{j\in \mathbb{Z}}\frac{1}{2}\left(1-(-1)^{j}\right)\Lambda^j.\label{Swaveoperator}
\end{align}
\begin{proposition}\label{Prop:t1kflow}
The wave operators $S^{\pm}$, $W^{\pm}$ and wave functions $\Psi^\pm$ satisfy
\begin{align}
&S^+\cdot\Lambda^{-1}\cdot S^-=W^+\cdot\Lambda^{-1}\cdot W^-=\iota_{\Lambda^{-1}}
(\Lambda-\Lambda^{-1})^{-1},\label{S+operators}\\
&\pa_{\mathbf{t}_k}S^{+}=B_{k}S^{+}-S^{+}\Lambda^{k},\quad \pa_{\mathbf{t}_k}S^{-}=\Lambda^{k}S^{-}+S^{-}B_{k}^*,\label{WaveoperaLambda}\\
&\pa_{\mathbf{t}_k}W^{+}=B_{k}W^{+},\quad \pa_{\mathbf{t}_k}W^{-}=W^{-}B_{k}^*,\quad\pa_{\mathbf{t}_k}\Psi^\pm=B_{k}(\Psi^\pm),\label{woperator}
\end{align}
where $\iota_{\Lambda^{\pm1}}
{(\Lambda-\Lambda^{-1})^{-1}}$ means  expanding ${(\Lambda-\Lambda^{-1})^{-1}}$ in $\mathbb{C}((\Lambda^{\pm1}))$ and
\begin{eqnarray}
&&B_{k}=\left(\Big(S^+\Lambda^{k-1} S^-\Big)_{\geq 0}+\Big(S^+\Lambda^{k-1} S^-\Big)^*_{<0}\right)\cdot(\Lambda-\Lambda^{-1}).\label{Bkdefine}
\end{eqnarray}
Here $(\sum_{i\in\mathbb{Z}}a_i\Lambda^i)_{\geq k}=\sum_{i\geq k}a_i\Lambda^i$, $(\sum_{i\in\mathbb{Z}}a_i\Lambda^i)_{< k}=\sum_{i< k}a_i\Lambda^i$.
\end{proposition}
\begin{proof}
Firstly if let $\mathbf{t''}=\mathbf{t'}=\mathbf{t}$, by comparing    negative  powers of $\Lambda$ in both sides of \eqref{Swaveoperator}, we can get \eqref{S+operators}.
If apply $\pa_{\mathbf{t'_k}}$ in  both sides of \eqref{Swaveoperator} and set $\mathbf{t''}=\mathbf{t'}=\mathbf{t}$,
\begin{eqnarray}
\pa_{\mathbf{t}_k}S^+\cdot\Lambda^{-1}\cdot S^-+S^+\cdot\Lambda^{k-1}\cdot S^-
=-\pa_{\mathbf{t}_k}(S^-)^*\cdot\Lambda \cdot (S^+)^*+(S^-)^*\cdot\Lambda^{1-k}\cdot (S^+)^*,\label{Sflow}
\end{eqnarray}
then by considering the negative  and  positive powers of $\Lambda$ in (\ref{Sflow}) respectively, one can obtain
\begin{eqnarray*}
&& \pa_{\mathbf{t}_k}S^+=\left(-S^+\cdot\Lambda^{k-1}\cdot S^-+(S^-)^*\cdot\Lambda^{1-k}\cdot (S^+)^*)\right)_{<0}\cdot(\Lambda-\Lambda^{-1})\cdot S^+,\\
&&\pa_{\mathbf{t}_k}(S^-)^*=\left(-S^+\cdot\Lambda^{k-1}\cdot S^-+(S^-)^*\cdot\Lambda^{1-k}\cdot (S^+)^*\right)_{>0}\cdot(\Lambda^{-1}-\Lambda)\cdot(S^-)^*,
\end{eqnarray*}
where we have used   (\ref{S+operators}). Then based upon this, \eqref{WaveoperaLambda} can be easily obtained. As for \eqref{woperator}, it can be easily proved by \eqref{waveoperator}, \eqref{Wavefunctions} and \eqref{WaveoperaLambda}.
\end{proof}
\begin{remark}
According to \eqref{Bkdefine}, one can find
\begin{align}
B_k^*\cdot \left(\Lambda-\Lambda^{-1}\right)=-(\Lambda-\Lambda^{-1})\cdot B_{k},\label{Bkconjuctur}
\end{align}
which guarantees that constraint \eqref{S+operators} is consistent with  \eqref{WaveoperaLambda} and \eqref{woperator}, that is, the constraint \eqref{S+operators} is invariant under the time flow $\mathbf{t}_k$. Further
 it can be found that $B_k$ has the following form
\begin{align}
B_k=\left(a_{0k}+\sum_{j=1}^{k-1}(a_{jk}\Lambda^j
+\Lambda^{-j}a_{jk})\right)(\Lambda-\Lambda^{-1}),\label{BKseconddefine}
\end{align}
which is given in \cite{Krichever2023}. But {\bf its connections with Lax operators are not clear}, so next we will try to express $B_k$ in terms of Lax operators.
\end{remark}
If introduce  Lax operators $L^{\pm}$ as follows
\begin{align}
L^+=S^+\cdot\Lambda\cdot (S^+)^{-1}=\sum_{j=-1}^\infty u^+_j\Lambda^{-j},\quad
L^-=(S^-)^*\cdot\Lambda^{-1}\cdot \left((S^-)^*\right)^{-1}=\sum_{j=-1}^\infty u^-_j\Lambda^{j}\label{Llaxdefine},
\end{align}
then according to \eqref{WaveoperaLambda}, one can obtain
\begin{eqnarray}
&&(L^+)^*=(\Lambda-\Lambda^{-1})\cdot L^{-}\cdot\iota_{\Lambda}{\left(\Lambda-\Lambda^{-1}\right)^{-1}},\ \pa_{\mathbf{t}_k}L^{\pm}=[B_{k},L^{\pm}].\label{L+adjoint}
\end{eqnarray}
Note that $B_k$ is expressed by wave operators $S^{\pm}$. By the lemma below, one can express  $B_k$ by $L^{\pm}$.
\begin{lemma}\label{Lemma:OperatorB}
Given two pseudo--difference operators $A=\sum_{j=-k}^{+\infty}a_j(n,\mathbf{t})\Lambda^{-j}$ and $B=\sum_{j=-k}^{+\infty}b_j(n,\mathbf{t})\Lambda^{j}$ $(k\geq1)$ satisfying
\begin{align}
A^{*}\cdot(\Lambda-\Lambda^{-1})=(\Lambda-\Lambda^{-1})\cdot B,\label{ABconj}
\end{align}
if set $A=\tilde{A}(\Lambda-\Lambda^{-1})$ and $B=\tilde{B}(\Lambda^{-1}-\Lambda)$,
then
$
\tilde{A}^*=\tilde{B}
$
and
\begin{align}
\left(\tilde{A}_{\geq0}+\tilde{B}_{<0}\right)(\Lambda-\Lambda^{-1})
=A_{\geq1}-B_{\leq-1}-
\left(A_{\geq1}-B_{\leq-1}\right)|_{\Lambda=1}.\label{Coperat}
\end{align}
 \end{lemma}
 \begin{proof}
 Firstly by \eqref{ABconj}, we can get $\tilde{A}^*=\tilde{B}$. So $\tilde{A}^*_{[0]}=\tilde{B}_{[0]}$, where $(\sum_{i\in\mathbb{Z}}a_i\Lambda^i)_{[0]}=a_0$.
If denote
 \begin{align}
 P=\left(\tilde{A}_{\geq0}+\tilde{B}_{<0}\right)(\Lambda-\Lambda^{-1}),\label{Pdefine}
 \end{align}
then due to $\left(\tilde{B}_{<0}\cdot (\Lambda-\Lambda^{-1})\right)_{\geq 1}=0$, we can find
 \begin{align}
 P_{\geq1}&=\left(\Big(\tilde{A}_{\geq0}+\tilde{B}_{<0}\Big)(\Lambda-\Lambda^{-1})\right)_{\geq1}
 =\left((\tilde{A}_{\geq0}+\tilde{A}_{<0})\cdot (\Lambda-\Lambda^{-1})\right)_{\geq1}=\left(\tilde{A}_{\geq0}\cdot (\Lambda-\Lambda^{-1})\right)_{\geq1}=A_{\geq1}.\label{p1define}
 \end{align}
 Similarly by $\left(\tilde{A}_{>0}\cdot (\Lambda-\Lambda^{-1})\right)_{\leq -1}=0$, we can get
  \begin{align}
 P_{\leq-1}=-B_{\leq-1}.\label{p-1define}
 \end{align}
 Next by \eqref{Pdefine}, we can find the actions of $P$ on $1$ vanish, i.e., $P\left(1\right)=0$. Therefore
\begin{align}
(P)_{[0]}=-\left(\tilde{A}_{\geq0}+\tilde{B}_{<0}\right)(\Lambda-\Lambda^{-1})
|_{\Lambda=1}=-\left({A}_{\geq1}-{B}_{\leq-1}\right)
|_{\Lambda=1}.\label{P0define}
\end{align}
 Finally  \eqref{Coperat} can be proved by summarizing \eqref{p1define}--\eqref{P0define}.
 \end{proof}
 \begin{theorem}\label{Propo:Bkdefine}
 In terms of Lax operators $L^{\pm}$,
\begin{eqnarray}\label{b1kandlax}
B_{k}=(L^{+})^k_{\geq 1}-(L^{-})^k_{\leq -1}-\left((L^{+})^k_{\geq 1}-(L^{-})^k_{\leq -1}\right)|_{\Lambda=1}.
\end{eqnarray}
\end{theorem}
\begin{proof}
Firstly by \eqref{L+adjoint}, we can get
\begin{align*}
\left((L^{+})^k\right)^*(\Lambda-\Lambda^{-1})=(\Lambda-\Lambda^{-1})(L^-)^k.
\end{align*}
Then by \eqref{S+operators} and \eqref{Llaxdefine},
\begin{align*}
(L^{+})^k=S^+\Lambda^{k-1}S^-(\Lambda-\Lambda^{-1}),\quad
(L^{-})^k=(S^-)^*\Lambda^{-k+1}(S^+)^*(\Lambda^{-1}-\Lambda).
\end{align*}
Finally according to
\eqref{Bkdefine} and Lemma \ref{Lemma:OperatorB}, this theorem can be easily proved.
\end{proof}
\begin{remark}\label{BkdefineDelta}
Here \eqref{Llaxdefine}, \eqref{L+adjoint} and \eqref{b1kandlax} are the Lax equations of the large BKP hierarchy. By the constraint  \eqref{L+adjoint}, we can find the initial terms in $L^{\pm}$ coincide, that is
\begin{align*}
 u_{-1}^+=u_{-1}^-,
 \end{align*}
 which shows that the large BKP hierarchy is just the Toda  hierarchy with constraint of type $B$ in \cite{Krichever2023}. If denote $u=u_{-1}^+=u_{-1}^-$ and introduce $\Delta=\Lambda-1$, $\Delta^*=\Lambda^{-1}-1$, then
\begin{align}
L^+=u\Delta+\sum_{j=0}^{+\infty}\tilde{u}_j^{+}\Delta^{-j},\quad
L^-=u\Delta^*+\sum_{j=0}^{+\infty}\tilde{u}_j^{-}(\Delta^*)^{-j},\label{reLaxdefine}
\end{align}
where we have used $\Lambda^{-l}=(1+\Delta)^{-l}=\sum_{j=0}^{+\infty}\binom{-l}{j} \Delta^{-j}$ and $\Lambda^{l}=(1+\Delta^*)^{-l}=\sum_{j=0}^{+\infty}\binom{-l}{j} (\Delta^*)^{-j}$.
The constraint in \eqref{L+adjoint} becomes
\begin{align*}
(\Delta-\Delta^*)(L^{+})^{*}=(\Delta-\Delta^*)L^{-}.
\end{align*}
Note that $(L^{+})^k_{\geq 1}-(L^{+})^k_{\geq 1}|_{\Lambda=1}=P^+_{\geq0}\cdot\Delta$ and
$(L^{-})^k_{\leq -1}-(L^{-})^k_{\leq -1}|_{\Lambda=1}=P^-_{\geq0}\cdot\Delta^*$, so we can get
\begin{align*}
B_k=(L^+)^k_{\Delta,\geq1}-(L^-)^k_{\Delta^*,\geq1}.
\end{align*}
Here $\sum_{i\in\mathbb{Z}}(a_iP^i)_{P, \geq k}=\sum_{i\geq k}a_iP^i$ with $P=\Delta$ or $\Delta^*$.
\end{remark}
  \begin{proposition}
$B_k$ satisfies
\begin{eqnarray}
&&\pa_{\mathbf{t}_k}B_l-\pa_{\mathbf{t}_l}B_k+[B_{l},B_{k}]=0,\label{paklaxandblax}
\end{eqnarray}
which implies
$[\partial_{\mathbf{t}_k}, \partial_{\mathbf{t}_l}]=0.$
\end{proposition}
\begin{proof}
Firstly denote $P_{kl}=\pa_{\mathbf{t}_k}B_l-\pa_{\mathbf{t}_l}B_k+[B_{l},B_{k}],$ then by \eqref{L+adjoint} and \eqref{p1define}
\begin{align*}
(P_{kl})_{\geq 1}=[B_k,(L^+)^l]_{\geq 1}-[B_l,(L^+)^k]_{\geq 1}+[B_l,B_k]_{\geq 1}.
\end{align*}
Note that  for two pseudo--difference operators A and B \begin{align}
[A,B]_{\geq1}=[A_{\geq1},B]_{\geq1}
+[A_{\leq0},B_{\geq1}]_{\geq1},\label{AoperatorA}
\end{align}
then by $(L^+)^l-B_l=(L^+)^l_{\leq0}+(L^-)^l_{\leq-1}-{B_l}_{[0]},$ we can get
\begin{small}\begin{align*}
\left[B_k,(L^+)^l-B_{l}\right]_{\geq1}
=\left[(L^+)^k_{\geq 1}, (L^+)^l-B_l\right]_{\geq 1}.
\end{align*}
\end{small}
So based upon above, one can find by  \eqref{p1define} and \eqref{AoperatorA}
\begin{small}\begin{align*}
(P_{kl})_{\geq 1}&=
\left[(L^+)^k_{\geq 1}, (L^+)^l-B_l\right]_{\geq 1}-\left[{B_l}, (L^+)^k\right]_{\geq 1}
=\left[(L^+)^k_{\geq1}, (L^+)^l\right]_{\geq 1}
+\left[(L^+)^k_{\leq 0}, B_l\right]_{\geq 1}\\
&=\left[(L^+)^k_{\geq1}, (L^+)^l\right]_{\geq 1}
+\left[(L^+)^k_{\leq 0},(L^+)^l\right]_{\geq 1}
=\left[(L^+)^k, (L^+)^l\right]_{\geq 1}=0,
\end{align*}
\end{small}
where have used  $[A_{\leq 0}, B_{\leq 0}]_{\geq 1}=0$.
Similarly one can prove $(P_{kl})_{\leq -1}=0.$

Note that by \eqref{b1kandlax}, one can find  $B_k\left(1\right)=0$, which means the sum of coefficients of $\Lambda$ in $B_k$ is zero. So we can get $$\pa_{\mathbf{t}_k}B_l\left(1\right)
=\pa_{\mathbf{t}_l}B_k\left(1\right)
=[B_k,B_l]\left(1\right)=0,$$
then
\begin{align*}
(P_{kl})_{[0]}=-\Big((P_{kl})_{\geq1}-(P_{kl})_{\leq-1}\Big)(1)=0.
\end{align*}
Summarize above results, we can finally get \eqref{paklaxandblax}. As for $[\partial_{\mathbf{t}_k}, \partial_{\mathbf{t}_l}]=0,$ it can be  proved by \eqref{L+adjoint} and \eqref{paklaxandblax}.
\end{proof}
\noindent {\bf Example:}
By the constraints on Lax operators $L^{\pm}$ $(see\ \eqref{L+adjoint})$, it can be found that all  $u^-_{j}$ in Lax operator $L^-$ can be expressed by  $u^+_{j}$ in $L^+$, that is
\begin{align*}
 &u^-_{-1}(n)= u^+_{-1}(n)\overset{\bigtriangleup}{=}u(n),\qquad u^-_{0}(n)=u^+_{0}(n+1),\\
 &u^-_{j}(n-j)- u^-_{j-2}(n-j+2)= u^+_{j}(n+1)- u^+_{j-2}(n-1),\quad j=1,2,\cdots,
\end{align*}
where  $u^{\pm}_{j}(n)$ means  $u^{\pm}_{j}(n,\mathbf{t})$.
Next we give some examples of $B_k$
\begin{align*}
\bullet\ &B_1=u(n)(\Lambda-\Lambda^{-1}),\\
\bullet\ &B_2=\left(u(n)\big(u^+_{0}(n+1)+u^+_{0}(n)\big)
+u(n)u(n+1)\Lambda
+\Lambda^{-1}u(n)u(n+1)\right)(\Lambda-\Lambda^{-1}).
\end{align*}
By  comparing coefficients of $\Lambda^i$ in both sides of Lax equation \eqref{L+adjoint}, one can get
\begin{align}
u(n)_{t_1}&=u(n)\left(
u^+_{0}(n+1)-u^+_{0}(n))\right),\label{u-1operator}\\
u^+_{0}(n)_{t_1}&=u(n)\left(u^+_{1}(n+1)
-u(n-1)+u(n+1)\right)-u^+_{1}(n)
u(n-1),\label{u01operator}\\
u(n)_{t_2}&=u(n)u(n+1)\left(u^+_{1}(n+2)
+u(n+2)-2u(n)\right)\nonumber\\
&\quad+u(n)
\left(u^+_{0}(n+1)+u^+_{0}(n)\right)
\left(u^+_{0}(n+1)-u^+_{0}(n)\right)\nonumber\\
&\qquad+u(n)u(n-1)\left(u(n)
-u^+_{1}(n)\right).\label{u+1operator}
\end{align}
If  set
 \begin{align*}
 f_0(n)=u(n)\left(u_0^+(n+1)+u^+_0(n)\right),
 \end{align*}
and  eliminate $u^+_1$ in \eqref{u-1operator}--\eqref{u+1operator}, we can get the integrable discretization of the Novikov--Veselov equation \cite{Krichever2023,Grushevsky2010}.
\begin{align*}
&\pa_{\mathbf{t_1}}\left({\rm log}\Big(u(n)u(n+1)\Big)\right)=\frac{f_0(n+1)}{u(n+1)}-
\frac{f_0(n)}{u(n)}, \\
&\pa_{\mathbf{t_2}}u(n)-\pa_{\mathbf{t_1}}f_0(n)
+2{u(n)}^2\left({u(n+1)}-
{u(n-1)}\right)=0.
\end{align*}

\subsection{ From  Lax equation   to  bilinear equation }
In this subsection, we will derive  bilinear equation \eqref{Hirotabkp} from Lax equation \eqref{L+adjoint}, which will be done in the following steps:
\begin{center}
\fbox{\textit{
{{\rm Lax operator}}}}$\longrightarrow $\fbox{\textit{
{{\rm wave operator}}}} $\longrightarrow $ \fbox{\textit{
{{\rm wave function}}}}$\longrightarrow $\fbox{\textit{
{{\rm tau function}}
}}
\end{center}

\subsubsection{\bf{From Lax equation to wave operators}} Firstly the Lax equation of the large BKP hierarchy is defined by
\begin{align}
&\pa_{\mathbf{t}_k}L^{\pm}=[B_{k},L^{\pm}],\label{laxequ2}\\
&B_{k}=(L^{+})^k_{\geq 1}-(L^{-})^k_{\leq -1}-\left((L^{+})^k_{\geq 1}-(L^{-})^k_{\leq -1}\right)|_{\Lambda=1},\label{bkequ2}
\end{align}
with the  Lax operators
\begin{eqnarray*}
L^+=\sum_{j=-1}^\infty u^+_j\Lambda^{-j},\quad
L^-=\sum_{j=-1}^\infty u^-_j\Lambda^{j}\label{L-lax2},
\end{eqnarray*}
 satisfying  the following  constraint
 \begin{eqnarray}
(L^{+})^*=(\Lambda-\Lambda^{-1})\cdot L^{-}\cdot\iota_{\Lambda}{\left(\Lambda-\Lambda^{-1}\right)^{-1}}.\label{laxcons}
\end{eqnarray}
Here we would like to point out that the systems of \eqref{laxequ2}--\eqref{laxcons} are consistent due to
\begin{align}
B_k^*\cdot \left(\Lambda-\Lambda^{-1}\right)=-(\Lambda-\Lambda^{-1})\cdot B_{k},\label{Bkcon}
\end{align}
which can be derived by Lemma \ref{Lemma:OperatorB}.

Next we will introduce the wave operators from  the Lax equation by the following lemmas.
\begin{lemma}\label{Lemma:abfunction}
Given  $a(n)\neq0$ and $b(n)$, if
\begin{align}
f(n)=f(n+1)a(n)+b(n),\label{fnfunction}
\end{align}
then
\begin{align*}
f(n)=\left\{\begin{matrix}
 \Pi_{k=0}^{\infty}a(n-k-1)^{-1},  &b(n)=0, \\
  -\sum_{l=0}^{\infty}b(n-l-1)\Pi_{k=0}^{l}a(n-k-1)^{-1}, &b(n)\neq0.
\end{matrix}\right.
\end{align*}
\end{lemma}
\begin{proof}
Firstly if $b(n)=0$, then
$$\log f(n)-\log f(n+1)=\log a(n),$$
that is
\begin{align*}
(1-\Lambda)\left(\log f(n)\right)=\log a(n).
 \end{align*}
Then \begin{align*}
 \log f(n)=-\sum_{k=0}^{+\infty}\log a(n-k-1).
 \end{align*}
Here we expand $(1-\Lambda)^{-1}=-\sum_{k=0}^{+\infty}\Lambda^{-k-1}$.
 So we can get
 $$f(n)=\Pi_{k=0}^{+\infty}a(n-k-1)^{-1}.$$
While when $b(n)\neq0$, we can assume $$f(n)=c(n)\Pi_{k=0}^{\infty}a(n-k-1)^{-1}.$$
Then  by \eqref{fnfunction}, it can be found that
\begin{align*}
(1-\Lambda)\left(c(n)\right)=b(n)\Pi_{k=0}^{+\infty}a(n-k-1)^{-1},
\end{align*}
which means
\begin{align*}
c(n)=-\sum_{l=0}^{+\infty}b(n-l-1)\Pi_{k=l+1}^{+\infty}a(n-k-1)^{-1}.
\end{align*}
Therefore we can finally obtain
\begin{align*}
f(n)= -\sum_{l=0}^{\infty}b(n-l-1)\Pi_{k=0}^{l}a(n-k-1)^{-1}.
\end{align*}
\end{proof}
\begin{remark}
 Instead of $(1-\Lambda)^{-1}=-\sum_{k=0}^{+\infty}\Lambda^{-k-1}$, we can also expand $(1-\Lambda)^{-1}=\sum_{k=0}^{+\infty}\Lambda^{k}$. In this case, we can get another kinds of  solutions. No matter how, \eqref{fnfunction}  has solutions for given as $a(n)\neq0$ and $b(n)$.
\end{remark}

\begin{proposition}
Given Lax operators of  the large BKP hierarchy $L^{\pm}=\sum_{j=-1}^{+\infty} u^{\pm}_j(n,\mathbf{t})\Lambda^{-j}$ satisfying  \eqref{laxequ2}--\eqref{laxcons}, there exist
 operators $S^+=\sum_{j=0}^{+\infty} b_j(n,\mathbf{t})\Lambda^{-j}$ and $S^-=\sum_{j=0}^{+\infty} \Lambda^{-j}\bar{b}_j(n,\mathbf{t})$
satisfying
 \begin{align}
 &L^{+}=S^{+}\Lambda(S^{+})^{-1},\quad L^{-}=(S^{-})^{*}\Lambda^{-1}\left((S^{-})^{-1}\right)^{*},\nonumber\\
  & S^{+}\Lambda^{-1} S^{-}=\iota_{\Lambda^{-1}}(\Lambda-\Lambda^{-1})^{-1},\nonumber\\
 &\pa_{\mathbf{t}_k}S^{+}=B_{k}S^{+}-S^{+}\Lambda^{k}, \quad \pa_{\mathbf{t}_k}S^{-}=\Lambda^{k}S^{-}+S^{-}B_{k}^*,\nonumber\\
 &B_{k}=\left(\Big(S^+\Lambda^{k-1} S^-\Big)_{\geq 0}+\Big(S^+\Lambda^{k-1} S^-\Big)^*_{<0}\right)\cdot(\Lambda-\Lambda^{-1}),\nonumber
   \end{align}
  where $S^+$ and $(S^{-})^{*}$ can be up to the multiplications of $C^{\pm}=\sum_{j=0}^{+\infty}c_j^{\pm}\Lambda^{\mp j}$ on the right sides with $C^+(C^-)^{*}=1.$ Here $c_j^{\pm}$ is some constant without depending on $n$ and $\mathbf{t}$.
\end{proposition}
\begin{proof}
Firstly  there exists an operator $\tilde{S}^{+}=\sum_{j=0}^{+\infty}b_j(n,\mathbf{t})\Lambda^{-j}$. In fact   by comparing coefficients of $\Lambda^{-j+1}$ $(j\geq0)$ in $L^{+}\tilde{S}^{+}=\tilde{S}^{+}\Lambda,$ we can get
\begin{align*}
\sum_{l=0}^{j-1} u^+_{j-l-1}(n,\mathbf{t})b_l(n-j+l+1,\mathbf{t})
+u^+_{-1}(n,\mathbf{t})b_j(n+1,\mathbf{t})=
b_j(n,\mathbf{t}).
\end{align*}
By  Lemma \ref{Lemma:abfunction}, we can express $b_j(n,\mathbf{t})$ $(j\geq0)$ in terms of $u^+_l(n,\mathbf{t})$ $(l\geq-1)$. Therefore given $L^+$, there exists $\tilde{S}^{+}$ such that  $L^{+}=\tilde{S}^{+}\Lambda(\tilde{S}^{+})^{-1}$.
Then consider the initial value problem of the following system of partial differential equations
\begin{equation}\label{initialproblem}
\left\{
 \begin{array}{lr}
\pa_{\mathbf{t}_k}S=B_{k}S-S\Lambda^{k}, &\\
 S|_{\mathbf{t}=0}=\tilde{S}^{+}(0), &\\
 \end{array}
\right.
\end{equation}
where $S$ has the following form $S=\sum_{j=0}^{+\infty}a_j\Lambda^{-j}$ and $B_k$ is given by \eqref{bkequ2}. Then by \eqref{paklaxandblax}, we can find
 \begin{align*}
 \pa_{\mathbf{t}_k}(B_lS-S\Lambda^l)= \pa_{\mathbf{t}_l}(B_kS-S\Lambda^k),
 \end{align*}
 which means that \eqref{initialproblem} has a   unique solution $\hat{{S}}^+$. Next we will show  $L^+=\hat{{S}}^+\Lambda(\hat{{S}}^+)^{-1}$. In fact by $\pa_{\mathbf{t}_k}\hat{{S}}^+=B_k\hat{{S}}^+-\hat{{S}}^+\Lambda^k,$ we can find
$L^+\hat{{S}}^+\Lambda^{-1}-\hat{{S}}^+=\sum_{j=0}^{+\infty}v_j\Lambda^{-j}$ satisfies
\begin{equation*}
\left\{
 \begin{array}{lr}
\pa_{\mathbf{t}_k}{S}
=B_k{S}-{S}\Lambda^k, &\\
{S}|_{\mathbf{t}=0}
=0, &\\
 \end{array}
\right.
\end{equation*}
where one should note that $(L^+{\hat{{S}}}^+\Lambda^{-1}-\hat{{S}}^+)\mid_{\mathbf{t}=0}
=L^+(0)\tilde{S}^+(0)\Lambda^{-1}-\tilde{S}^+(0)=0$.

Then if set
\begin{equation*}
\hat{{S}}^-=\Lambda(\hat{{S}}^+)^{-1}\iota_{\Lambda^{-1}}(\Lambda-\Lambda^{-1})^{-1},
\end{equation*}
by $L^+=\hat{{S}}^+\Lambda(\hat{{S}}^+)^{-1}$ and \eqref{laxcons}, we can find
\begin{align*}
(\hat{{S}}^-)^*\Lambda^{-1}((\hat{{S}}^-)^*)^{-1}=\iota_{\Lambda}(\Lambda-\Lambda^{-1})^{-1}(L^+)^*
(\Lambda-\Lambda^{-1})=L^-,
\end{align*}
and $\pa_{\mathbf{t}_k}\hat{{S}}^-$ can be derived by $\pa_{\mathbf{t}_k}\hat{{S}}^+=B_k\hat{{S}}^+-\hat{{S}}^+\Lambda^k$ and \eqref{Bkcon}, that is
$$\pa_{\mathbf{t}_k}\hat{{S}}^-=\Lambda^k\hat{{S}}^-+\hat{{S}}^-B_k^{*}.$$

Finally assume there exist $S_1^{\pm}$ and $S_2^{\pm}$ such that
\begin{align}
&L^+=S_i^+\Lambda(S_i^+)^{-1}=(S_i^-)^*\Lambda^{-1}((S_i^-)^*)^{-1},\quad i=1,2,\label{Laxoper}\\
&S_1^+\Lambda^{-1}S_1^-
=\iota_{\Lambda^{-1}}(\Lambda-\Lambda^{-1})^{-1}
=S_2^+\Lambda^{-1}S_2^-.\nonumber
\end{align}
Then by \eqref{Laxoper},
\begin{align}
[(S_1^+)^{-1}S_2^+,\Lambda]=[((S_1^-)^{*})^{-1}(S_2^-)^*,\Lambda^{-1}]=0.\label{operatorS}
\end{align}
So if set
\begin{align*}
C^+=\sum_{j=0}^{+\infty}c_j^+(n,\mathbf{t})\Lambda^{-j}=(S_1^+)^{-1}S_2^+,\quad
C^-=\sum_{j=0}^{+\infty}c_j^-(n,\mathbf{t})\Lambda^{j}=((S_1^-)^{*})^{-1}(S_2^-)^*,
\end{align*}
then by \eqref{operatorS}, we can find $c_j^{\pm}(n,\mathbf{t})$ does not depend on $n$. So we can write $c_j^{\pm}(\mathbf{t})$ instead of $c_j^{\pm}(n,\mathbf{t})$. In this case
\begin{equation}\label{SCoperator}
{S}_2^+={S}_1^+ C^+,\quad ({S}_2^-)^{*}=({S}_1^-)^{*}C^-.
\end{equation}
If insert \eqref{SCoperator} into $S_1^+\Lambda^{-1}S_1^-
=S_2^+\Lambda^{-1}S_2^-$, we can get $C^+(C^-)^*=1$.
Further by $\pa_{\mathbf{t}_k} S_i^+=B_kS^+_i-S_i^+\Lambda^k$ and  the fact that $C^+$ does not depend on $n$, we can obtain
\begin{align*}
\pa_{\mathbf{t}_k}C^+&=-({S}_1^+)^{-1}\pa_{\mathbf{t}_k}{S}_1^+({S}_1^+)^{-1}{S}_2^+
+({S}_1^+)^{-1}\pa_{\mathbf{t}_k}{S}_2^+\\
&=-({S}_1^+)^{-1}(B_kS^+_1-S_1^+\Lambda^k)({S}_1^+)^{-1}{S}_2^+
+({S}_1^+)^{-1}(B_kS^+_2-S_2^+\Lambda^k)=[\Lambda^k,({S}_1^+)^{-1}{S}_2^+]=
0,
\end{align*}
which means $c_j^+(n,\mathbf{t})$ is some constant without depending $n$ and $\mathbf{t}$. Finally by  $C^+(C^-)^*=1$, we can know  $c_j^-(n,\mathbf{t})$ also does not depend on $n$ and $\mathbf{t}$. At last we can express $B_k$ in the forms of $S^{\pm}$ by \eqref{bkequ2} and Lemma \ref{Lemma:OperatorB}.
\end{proof}
\subsubsection{\bf{From wave operators to wave functions}}
In this subsection, we will start from the wave operators   of the large BKP hierarchy,  which is defined by the wave operators
\begin{eqnarray*}
S^+(n,\mathbf{t},\Lambda)=\sum_{j=0}^\infty b_j(n,\mathbf{t})\Lambda^{-j},\quad S^-(n,\mathbf{t},\Lambda)=\sum_{j=0}^\infty \Lambda^{-j}\bar{b}_j(n,\mathbf{t}),
\end{eqnarray*}
satisfying
\begin{align*}
&S^+\cdot\Lambda^{-1}\cdot S^-=\iota_{\Lambda^{-1}}
(\Lambda-\Lambda^{-1})^{-1},\nonumber\\
&\pa_{k}S^{+}=B_{k}S^{+}-S^{+}\Lambda^{k},\quad \pa_{k}S^{-}=\Lambda^{k}S^{-}+S^{-}B_{k}^*.
\end{align*}
Here
\begin{align}
&B_{k}=\left(\Big(S^+\Lambda^{k-1} S^-\Big)_{\geq 0}+\Big(S^+\Lambda^{k-1} S^-\Big)^*_{<0}\right)\cdot(\Lambda-\Lambda^{-1}).\label{Bkdefine2}
\end{align}

If further set
\begin{eqnarray*}
W^+(n,\mathbf{t},\Lambda)
=S^+(n,\mathbf{t},\Lambda)e^{\xi(\mathbf{t},\Lambda)},\quad W^-(n,\mathbf{t},\Lambda)
=e^{-\xi(\mathbf{t},\Lambda)}S^-(n,\mathbf{t},\Lambda),
\end{eqnarray*}
then wave operators  $W^{\pm}$  satisfy
\begin{align}
&W^+\cdot\Lambda^{-1}\cdot W^-=\iota_{\Lambda^{-1}}
(\Lambda-\Lambda^{-1})^{-1},\label{S+operators2}\\
&\pa_{k}W^{+}=B_{k}W^{+},\quad \pa_{k}W^{-}=W^{-}B_{k}^*.\label{woperator2}
\end{align}
\begin{proposition}\label{bilinasr}
Given operators $W^{\pm}$  satisfying \eqref{S+operators2} and \eqref{woperator2}
\begin{align}
&& W^+(n,\mathbf{t},\Lambda)\cdot\Lambda^{-1}\cdot \pa^\alpha W^-(n,\mathbf{t},\Lambda)+\left(\pa^\alpha W^+(n,\mathbf{t},\Lambda)\cdot\Lambda^{-1}\cdot W^-(n,\mathbf{t},\Lambda)\right)^*=\left\{\begin{matrix}
 &0, & \alpha \ne 0,\\
 &\frac{1}{2}\sum_{j\in \mathbb{Z}}\left(1-(-1)^j\right)\Lambda^j, & \alpha =0,
\end{matrix}\right.\label{Wavefunctiontaylor}
\end{align}
where $\alpha=(\alpha_{1},\alpha_{2},\cdots)\in\mathbb{Z}_+$, $\pa^\alpha=\prod_{k=1}^\infty\pa_{k}^{\alpha_{k}}$.
If set
\begin{align*}
\Psi^{+}(n,\mathbf{t},z)=W^+(n,\mathbf{t},\Lambda)(z^{n}),\quad
\Psi^{-}(n,\mathbf{t},z)
=\left(W^-(n,\mathbf{t},\Lambda)\right)^*(z^{-n-1}),
\end{align*}
then
\begin{align}
{\rm Res}_{z}{z^{-1}}\left(\Psi^+(n',\mathbf{t'},z)
\Psi^-(n'',\mathbf{t}'',z)+\Psi^-(n',\mathbf{t'},z)
\Psi^+(n'',\mathbf{t}'',z)\right)=\frac{1}{2}\left(1-(-1)^{n'-n''}\right). \label{bilinear function 2}
\end{align}

\end{proposition}
\begin{proof}
We will try to prove (\ref{Wavefunctiontaylor}) by  induction on  $|\alpha|=\sum_{j=1}^{+\infty} \alpha_{j}$.
Firstly when $|\alpha|=0$, (\ref{Wavefunctiontaylor}) is correct by  \eqref{S+operators2}, while for $|\alpha|=1$ (e.g. $\alpha_k=1$ and other $\alpha_j=0$ $(j\neq k)$), (\ref{Wavefunctiontaylor}) can be proved by \eqref{woperator2} and
 \begin{align*}
 \big(\iota_{\Lambda^{-1}}(\Lambda-\Lambda^{-1})^{-1}
+\iota_{\Lambda}(\Lambda^{-1}-\Lambda\big)^{-1}\big)B_k^*=0,
 \end{align*}
 derived via  \eqref{Bkdefine2}. Now let us assume (\ref{Wavefunctiontaylor}) is correct for $|\alpha|\geq1$. Then
 \begin{eqnarray*}
&W^+(n,\mathbf{t},\Lambda)\cdot\Lambda^{-1}\cdot \pa_{\mathbf{t}_l}\pa^{\alpha}W^-(n,\mathbf{t},\Lambda)+\left( \pa_{\mathbf{t}_l}\pa^{\alpha}W^+(n,\mathbf{t},\Lambda)\cdot\Lambda^{-1}\cdot W^-(n,\mathbf{t},\Lambda)\right)^*\\
&=\pa_{\mathbf{t}_l}\left(W^+(n,\mathbf{t},\Lambda)\cdot\Lambda^{-1}\cdot \pa^{\alpha}W^-(n,\mathbf{t},\Lambda)+\left( \pa^{\alpha}W^+(n,\mathbf{t},\Lambda)\cdot\Lambda^{-1}\cdot W^-(n,\mathbf{t},\Lambda)\right)^*\right)\\
&-\pa_{\mathbf{t}_l}\Big( W^+(m,\mathbf{t},\Lambda)\Big)\cdot\Lambda^{-1}\cdot\pa^{\alpha} W^-(n,\mathbf{t},\Lambda)-\pa_{\mathbf{t}_l}\Big( W^-(n,\mathbf{t},\Lambda)^*\Big)\cdot\Lambda\cdot\pa^{\alpha} W^+(n,\mathbf{t},\Lambda)^*\\
&=-B_l\left( W^+(n,\mathbf{t},\Lambda)\cdot\Lambda^{-1}\cdot\pa^{\alpha} W^-(n,\mathbf{t},\Lambda)+ W^-(n,\mathbf{t},\Lambda)^*\cdot\Lambda\cdot\pa^{\alpha} W^+(n,\mathbf{t},\Lambda)^*\right)=0,
\end{eqnarray*}
where we have used \eqref{woperator2} and the fact that (\ref{Wavefunctiontaylor}) holds for $|\alpha|$. So (\ref{Wavefunctiontaylor}) is also correct for $|\alpha|+1$.

 Next by Taylor expansion
 \begin{align*}
 f({\mathbf{t}'})=\sum_{\alpha=0}^{+\infty}\frac{\pa^{\alpha}f(\mathbf{t})}{\alpha!}(\mathbf{t}'-\mathbf{t})^{\alpha},
 \end{align*}
 with $\alpha!=\Pi_{j=1}^{+\infty}{\alpha_j}$ and $(\mathbf{t}'-\mathbf{t})^{\alpha}=\Pi_{j=1}^{+\infty} (\mathbf{t}_j'-\mathbf{t}_j)^{\alpha_j}$, we can get
 \begin{align*}
 &W^+(n,\mathbf{t'},\Lambda)\cdot\Lambda^{-1}\cdot W^-(n,\mathbf{t'}',\Lambda)+\left(W^+(n,\mathbf{t}'',\Lambda)
 \cdot\Lambda^{-1}\cdot W^-(n,\mathbf{t'},\Lambda)\right)^*=\sum_{j\in \mathbb{Z}}\frac{1}{2}\left(1-(-1)^{j}\right)\Lambda^j.
 \end{align*}
Finally by Lemma \ref{Lemma:ABlambda},  \eqref{bilinear function 2} can be easily proved based upon (\ref{Wavefunctiontaylor}).
\end{proof}
\subsubsection{\bf{The existence of the tau function}}
In this subsection, we will show  the existence of the tau function, starting  from the bilinear equation
\begin{align}
{\rm Res}_{z}{z^{-1}}\left(\Psi^+(n',\mathbf{t'},z)
\Psi^-(n'',\mathbf{t}'',z)+\Psi^-(n',\mathbf{t'},z)
\Psi^+(n'',\mathbf{t}'',z)\right)=\frac{1}{2}\left(1-(-1)^{n'-n''}\right), \label{bilinear function 3}
\end{align}
where wave functions $\Psi^{\pm}(n,\mathbf{t},z)$ have the following forms
 \begin{align}
\Psi^{+}(n,\mathbf{t},z)=\tilde{\psi}^{+}(n,\mathbf{t},z)e^{\xi(\mathbf{t},z)}z^n,\quad
\Psi^{-}(n,\mathbf{t},z)=\tilde{\psi}^{-}(n,\mathbf{t},z)e^{-\xi(\mathbf{t},z)}z^{-n-1},\label{phidefine}
\end{align}
 with $$\tilde{\psi}^{\pm}(n,\mathbf{t},z)
=\sum_{j=0}^{+\infty}w_j^{\pm}(n,\mathbf{t})z^{-j}.$$
\begin{proposition}\label{Prop:taubilinear}
For $\tilde{\psi}^{+}(n,\mathbf{t},z)$ in \eqref{phidefine},
\begin{eqnarray}
&&\tilde{\psi}^+(n,\mathbf{t}-[\lambda_1^{-1}],\lambda_2)
\tilde{\psi}^-(n-1,\mathbf{t}-[\lambda_2^{-1}],\lambda_2)(1-\lambda_2\lambda_1^{-1})=1,\label{psippsim13}\\
&&\tilde{\psi}^+(n,\mathbf{t}-[\lambda_1^{-1}],\lambda_2)
\tilde{\psi}^-(n,\mathbf{t}-[\lambda_2^{-1}],\lambda_2)=\tilde{\psi}^-(n,\mathbf{t}-[\lambda_1^{-1}],\lambda_1)
\tilde{\psi}^+(n,\mathbf{t}-[\lambda_2^{-1}],\lambda_1).\label{psippsim12}
\end{eqnarray}
In particular
\begin{eqnarray}
&&\tilde{\psi}^+(n,\mathbf{t},\lambda)
\tilde{\psi}^-(n-1,\mathbf{t}-[\lambda^{-1}],\lambda)=1.\label{psippsimlambda}
\end{eqnarray}
\end{proposition}
\begin{proof}
 Firstly we can get
 (\ref{psippsim13}) by setting $n''=n-1,$ $n'=n$, $\mathbf{t'}\rightarrow \mathbf{t}-[\lambda_1^{-1}]$ and $\mathbf{t'}'\rightarrow \mathbf{t}-[\lambda_2^{-1}]$ in (\ref{bilinear function 3}),
while \eqref{psippsim12} can be obtained by  $n''=n'=n$, $\mathbf{t'}\rightarrow \mathbf{t}-[\lambda_1^{-1}]$, $\mathbf{t}''\rightarrow \mathbf{t}-[\lambda_2^{-1}]$ in (\ref{bilinear function 3}).
At last, \eqref{psippsimlambda} comes from by letting $\lambda_1\rightarrow\infty,$ $\lambda_2\rightarrow\lambda$ in \eqref{psippsim13}.
\end{proof}
If define \begin{align}
\Omega(n,\mathbf{t})=\sum_{k=1}^{\infty}{\rm Res }_zz^{k}N(z){\rm log}\tilde{\psi}^+(n,\mathbf{t},z)d{\mathbf{t}}_k,\label{psidefine}
\end{align}
where $N(z)=-\sum_{j=1}^{+\infty}z^{-j-1}\pa{t_j}+\pa_z$, then we can obtain the following lemma.
\begin{lemma}
\begin{align}
&\Omega(n+1,\mathbf{t})-e^{-\xi(\tilde{\pa},z^{-1})}\Omega(n,\mathbf{t})
=-d{\rm log}\tilde{\psi}^+(n+1,\mathbf{t},z).\label{Obilinaea1}
\end{align}
\end{lemma}
\begin{proof}
Firstly by \eqref{psippsimlambda}, it can be found that \eqref{psippsim12}  becomes
\begin{align}
&e^{-\xi(\tilde{\pa},\lambda_1^{-1})}{\rm log}\tilde{\psi}^+(n,\mathbf{t},\lambda_2)-
{\rm log}\tilde{\psi}^+(n+1,\mathbf{t},\lambda_2)
=e^{-\xi(\tilde{\pa},\lambda_2^{-1})}{\rm log}\tilde{\psi}^+(n,\mathbf{t},\lambda_1)-
{\rm log}{\psi}^+(n+1,\mathbf{t},\lambda_1).\label{bilinear1}
\end{align}
By applying $N(\lambda_1)$ on \eqref{bilinear1}, we can get
\begin{align*}
N(\lambda_1){\rm log}\tilde{\psi}^+(n+1,\mathbf{t},\lambda_2)
=N(\lambda_1){\rm log}\tilde{\psi}^+(n+1,\mathbf{t},\lambda_1)
-N(\lambda_1)e^{-\xi(\tilde{\pa},\lambda_2^{-1})}{\rm log}\tilde{\psi}^+(n,\mathbf{t},\lambda_1).
\end{align*}
Then  by definition of $\Omega$, we can obtain \eqref{Obilinaea1}.
\end{proof}
\begin{theorem}\label{Theorem:taufuncyion}
For $\Psi^{\pm}(n,\mathbf{t},z)$ in \eqref{phidefine}, there exists a tau function $\tau_n(\mathbf{t})$ of the large BKP hierarchy satisfying
\begin{align*}
\Psi^{+}(n,\mathbf{t},z)
=\frac{\tau_{n-1}(\mathbf{t}-[z^{-1}])}{\tau_n(\mathbf{t})}
e^{\xi(\mathbf{t},z)}z^{n},\quad
\Psi^{-}(n,\mathbf{t},z)
=\frac{\tau_{n+1}(\mathbf{t}+[z^{-1}])}{\tau_n(\mathbf{t})}
e^{-\xi(\mathbf{t},z)}z^{-n-1}.
\end{align*}
\end{theorem}
\begin{proof}
Firstly if set $z\rightarrow\infty$ in \eqref{Obilinaea1}, we can get
\begin{align}
&\Omega(n+1,\mathbf{t})-\Omega(n,\mathbf{t})
=-d{\rm log}w_0^+(n+1,\mathbf{t}).\label{bilinaea12}
\end{align}
Next after applying $d$  to both sides of \eqref{Obilinaea1} and \eqref{bilinaea12}, then we can get
\begin{align*}
&d\Omega(n+1,\mathbf{t})
=e^{-\xi(\tilde{\pa},z^{-1})}d\Omega(n,\mathbf{t}),\quad
d\Omega(n,\mathbf{t})=d\Omega(n+1,\mathbf{t}),
\end{align*}
which implies
\begin{align*}
d\Omega(n,\mathbf{t})
=e^{-\xi(\tilde{\pa},z^{-1})}d\Omega(n,\mathbf{t}).
\end{align*}
So we can assume
\begin{align}
d\Omega(n,\mathbf{t})=\sum_{i,j=1}^{\infty}a_{ij}d\mathbf{t}_i\wedge d\mathbf{t}_j,\label{omegacondefine}
\end{align}
where  $a_{ij}=-a_{ji}$ does not depend on $n$ and $\mathbf{t}$.
Further by integration on $\mathbf{t}$, we can get
\begin{align}
\Omega(n,\mathbf{t})=\sum_{i,j=1}^{\infty}a_{ij}\mathbf{t}_i d\mathbf{t}_j+dg(n,\mathbf{t}),\label{coefficient}
\end{align}
for some functions $g(n,\mathbf{t})$.

Further by \eqref{coefficient}, it can be found that \eqref{Obilinaea1} becomes
\begin{align*}
&d{\rm log}\tilde{\psi}^+(n+1,\mathbf{t},z)=e^{-\xi(\tilde{\pa},z^{-1})}dg(n,\mathbf{t})
-dg(n+1,\mathbf{t})-\sum_{j=1}^{\infty}(\sum_{i=1}^{\infty}a_{ij}\frac{z^{-i}}{i}
)d\mathbf{t}_j,
\end{align*}
which implies that
\begin{align}
&{\rm log}\tilde{\psi}^+(n+1,\mathbf{t},z)=e^{-\xi(\tilde{\pa},z^{-1})}g(n,\mathbf{t})
-g(n+1,\mathbf{t})-\sum_{j=1}^{\infty}(\sum_{i=1}^{\infty}a_{ij}\frac{z^{-i}}{i}
)\mathbf{t}_j+A(n,z),\label{tau1}
\end{align}
with some functions $A(n,z)=\sum_{k=1}^{\infty}A_k(n,\mathbf{t})z^{-k}$.
Substituting \eqref{tau1} back into  \eqref{bilinear1}, then we can get
\begin{align*}
\sum_{j=1}^{\infty}\left(\sum_{i=1}^{\infty}a_{ij}\frac{z_1^{-i}}{i}\right)
\frac{z_2^{-j}}{j}
=\sum_{j=1}^{\infty}\left(\sum_{i=1}^{\infty}a_{ij}\frac{z_2^{-i}}{i}\right)
\frac{z_1^{-j}}{j},
\end{align*}
which implies  $a_{ij}=a_{ji}.$ But $a_{ij}=-a_{ji}$, thus $a_{ij}=0$. Now \eqref{omegacondefine} and  \eqref{coefficient} will become
\begin{align}
d\Omega(n,\mathbf{t})=0,\quad\Omega(n,\mathbf{t})=dg(n,\mathbf{t}).\label{omegadefine}
 \end{align}

So if let $g(n,\mathbf{t})={\rm log}\tau_n(\mathbf{t})$,
\eqref{tau1} will be
\begin{align}
&{\rm log}\tilde{\psi}^+(n+1,\mathbf{t},z)=e^{-\xi(\tilde{\pa},z^{-1})}{\rm log}\tau_n(\mathbf{t})
-{\rm log}\tau_{n+1}(\mathbf{t})+A(n,z).\label{omegadefine2}
\end{align}
If substitute above relations into the definition of $\Omega$ (see \eqref{psidefine}), we can find
 \begin{align*}
 \Omega(n+1,\mathbf{t})=d{\rm log}\tau_{n+1}(\mathbf{t})-\sum_{k=1}^{\infty}kA_k(n,\mathbf{t})d\mathbf{t}_k.
 \end{align*}
 So by comparing with \eqref{omegadefine}, we can get $A(n,z)=0$. Finally by \eqref{omegadefine2}, we can obtain
\begin{eqnarray}
&&\tilde\psi^+(n,\mathbf{t},z)
=\frac{\tau_{n{-}1}(\mathbf{t}-[z^{-1}])}{\tau_{n}(\mathbf{t})}.\label{tau1fun}
\end{eqnarray}
As for $\Psi^-(n,\mathbf{t},z)$, it can be obtained by \eqref{psippsimlambda} and \eqref{tau1fun}.
\end{proof}
\section{From  large BKP hierarchy to  Toda  hierarchy}\label{section4}
In this section, large BKP hierarchy is viewed as the sub--hierarchy of modified Toda  (mToda) hierarchy. Then we review  Miura links between mToda and Toda hierarchies. After that, we restrict  Miura links to the case of large BKP hierarchy, and show the relations of tau function with Toda hierarchy. Here one can refer to \cite{Rui2024} for more details about mToda hierarchies.
\subsection{Large BKP hierarchy as sub--hierarchy of mToda hierarchy }\label{largeBKPsubmToda}
Firstly the  mToda hierarchy \cite{Rui2024} is   defined by
 \begin{align}
\pa_{x_k}{L}_i=[({L}_1^k)_{\Delta,\geq1},{L}_i],\quad \pa_{y_k}{L}_i=[({L}^k_2)_{\Delta^*,\geq1},{L}_i],\quad i=1,2,\label{mTodalaxeq}
\end{align}
with Lax operators
\begin{align}
{L}_1=v(n)\Lambda +\sum_{j=0}^{\infty}v_j(n)\Lambda^{-j},\quad
{ {L}_2}=  \overline{v}(n)\Lambda^{-1} +\sum_{j=0}^{\infty}  \overline{v}_j(n)\Lambda^{j}.\label{mTodalaxdefine}
 \end{align}
 Note that this hierarchy contains modified Toda lattice equation \cite{Hirota2004,Dai2003}, that is
 \begin{align*}
 v(n)_{y_1}=v(n)\big(\overline{v}(n)-\overline{v}(n+1)\big),\quad \overline{v}(n)_{x_1}=\overline{v}(n)\big({v}(n)-{v}(n-1)\big).
 \end{align*}
 If further set ${v}(n)=\pa_{x_1}\overline{\varphi}(n)$, $\overline{v}(n)=e^{\overline{\varphi}(n)-\overline{\varphi}(n-1)}$,
then we can get \cite{Hirota2004,Dai2003}
 \begin{align*}
\pa_{x_1}\pa_{y_1}\overline{\varphi}(n)
+\left(e^{\overline{\varphi}(n+1)-\overline{\varphi}(n)}
-e^{\overline{\varphi}(n)-\overline{\varphi}(n-1)}\right)\pa_{x_1}\overline{\varphi}(n)=0.
 \end{align*}

The mToda  hierarchy can also be expressed by wave operators ${S}_i$ and ${W}_i$ such that
\begin{align}
L_1={S}_1\Lambda{S}_1^{-1},\quad L_2={S}_2\Lambda^{-1}{S}_2^{-1}, \label{mTodawavelax}
\end{align}
where
\begin{align}
{S}_1=\sum_{k=0}^{\infty}c_k\Lambda^{-k},\quad
{S}_2=\sum_{k=0}^{\infty}\bar{c}_k\Lambda^{k}
,\quad c_0\neq0,\quad \bar{c}_0\neq0,\label{waveoperatordefine}
\end{align}
 satisfying
\begin{align}
&\pa_{x_k}{S}_1=-( L^k_1)_{\Delta,\leq0}{S}_1,\quad \pa_{y_k}{S}_2=-( {L}^k_2)_{\Delta^*,\leq0}{S}_2, \label{mTodaS1}\\
&\pa_{y_k}{S}_1=( {L}^k_2)_{\Delta^*,\geq1}{S}_1,\quad
\pa_{x_k}{S}_2=( {L}^k_1)_{\Delta,\geq1}{S}_2.\label{mTodaS2}
\end{align}

There is  a self--transformation for mToda hierarchy defined by
\begin{align*}
\pi_1: &({L}_i,{S}_i,x_k,y_k)\rightarrow
(\tilde{{L}}_i,\tilde{{S}}_i,\tilde{x}_k,\tilde{y}_k),\end{align*}
with
\begin{align*}
&\tilde{{L}}_1=\iota_{\Lambda^{-1}}\Delta^{-1}\cdot{L}^*_{2}\cdot\Delta,\quad
\tilde{{L}}_2=\iota_{\Lambda}\Delta^{-1}\cdot{L}^*_{1}\cdot\Delta,\\
&\tilde{{S}}_1=\iota_{\Lambda^{-1}}\Delta^{-1}\cdot({S}^{-1}_{2})^*\cdot\Lambda,\quad
\tilde{{S}}_2=-\iota_{\Lambda}\Delta^{-1}\cdot({S}^{-1}_{1})^*,\\
&\tilde{x}_k=-y_k,\quad\tilde{y}_k=
-x_k,
\end{align*}
satisfying
\begin{align*}
\pa_{\tilde{x}_k}{\tilde{{L}}_i}
=[({\tilde{{L}}^k_1})_{\Delta,\geq1},\tilde{{L}}_i],
\quad \pa_{\tilde{y}_k}{\tilde{{L}}_i}
=[({\tilde{{L}}^k_2})_{\Delta^*,\geq1},\tilde{{L}}_i].
\end{align*}
Here we have used for operator $A$, $$(\iota_{\Lambda}\Delta^{-1}\cdot A^*\cdot\Delta)_{\Delta^*,\geq 1}=\iota_{\Lambda^{\pm1}}\Delta^{-1}\cdot(A_{\Delta,\geq 1})^*\cdot\Delta,\quad
(\iota_{\Lambda^{-1}}\Delta^{-1}\cdot A^*\cdot\Delta)_{\Delta,\geq 1}=\iota_{\Lambda^{\pm1}}\Delta^{-1}\cdot(A_{\Delta^*,\geq 1})^*\cdot\Delta,$$
where one should note that $\iota_{\Lambda}\Delta^{-1}\cdot\Delta^*=\iota_{\Lambda^{-1}}\Delta^{-1}\cdot\Delta^*=-\Lambda^{-1}$.
Note that
$
{\rm Ad}\Lambda^l\Big(L_i(n)\Big)=\Lambda^l\cdot L_i(n)\cdot\Lambda^{-l}=L_i(n+l),$ ${\rm Ad}\Lambda^l\Big(S_i(n)\Big)=\Lambda^l S_i(n)\Lambda^{-l}=S_i(n+l),
$
therefore ${\rm Ad}\Lambda^l$ is  another self--transformation and  communicates with $\pi_1$. By definition,  we can get
${\rm Ad}\Lambda\circ{\pi_1}^2=1,$
so $$\pi_1^{-1}={\rm Ad}\Lambda\circ\pi_1.$$

If introduce wave function ${\Phi}_i$  and adjoint wave function ${\Phi}^*_i$ $(i=1,2)$ as follows
\begin{align}
&{\Phi}_1(n,x,y,z)={S}_1(n,x,y,z)e^{\xi(x,\Lambda)}(z^n),\quad
{\Phi_1}^*(n,x,y,z)=-\iota_{\Lambda}\Delta^{-1}\cdot({S}_1(n,x,y,z)^{-1})^*
e^{-\xi(x,\Lambda^{-1})}(z^{-n}),\label{Phidefine1}\\
&{\Phi}_2(n,x,y,z)={S}_2(n,x,y,z)
e^{\xi(y,\Lambda^{-1})}(z^n),\quad
{\Phi_2}^*(n,x,y,z)=\iota_{\Lambda^{-1}}\Delta^{-1}\cdot({S}_2(n,x,y,z)^{-1})^*
e^{-\xi(y,\Lambda)}(z^{-n}).
\label{Phidefine2}
\end{align}
Then ${\Phi}_i$  and  ${\Phi}^*_i$  satisfy the bilinear equation below \cite{Rui2024}
\begin{align}
\oint_{z=\infty}\frac{dz}{2\pi i z}\Phi_1(n,x,y,z)\Phi^*_1(n',x',y',z)
+\oint_{z=0}\frac{dz}{2\pi i z}\Phi_2(n,x,y,z)\Phi^*_2(n',x',y',z)=1.\label{modifiedTodabilinear}
\end{align}
Further there exists tau functions $\big(\tau_{0,n}(x,y),\tau_{1,n}(x,y)\big)$ such that
\begin{eqnarray}
\Phi_1(n,x,y,z)=\frac{\tau_{0,n}
(x-[z^{-1}],y)}
{\tau_{1,n}(x,y)}e^{\xi(x,z)}z^n,&
&{\Phi_2}(n,x,y,z)=\frac{\tau_{0,n+1}
(x,y-[z])}
{\tau_{1,n}(x,y)}
e^{\xi(y,z^{-1})}z^{n},\label{mTLtau}\\
{\Phi_1}^*(n,x,y,z)=\frac{\tau_{1,n}
(x+[z^{-1}],y)}
{\tau_{0,n}(x,y)}e^{-\xi(x,z)}z^{-n},&
&{\Phi_2}^*(n,x,y,z)=\frac{\tau_{1,n-1}
(x,y
+[z])}{\tau_{0,n}(x,y)}
e^{-\xi(y,z^{-1})}z^{-n+1}.\label{mTLtau2}
\end{eqnarray}
So we can find \eqref{modifiedTodabilinear} will become
\begin{align}\label{HirotamToda}
&\text{{\rm Res}$_{z}$}\left(z^{n-n'-1}\tau_{0, n}(x-[z^{-1}],y)
\cdot\tau_{1, n'}(x'+{[z^{-1}]},y)e^{\xi(x-x',z)}\right.\nonumber\\
&\quad\quad\quad\left.
+z^{n'-n-2}\tau_{0, n+1}(x,y-{[z^{-1}]})
\cdot\tau_{1,n'-1}(x',y'+[z^{-1}])e^{\xi(y-y',z)}\right)=\tau_{1,n}
(x,y)\tau_{0,n'}(x',y').
\end{align}

If add the constraint \cite{Krichever2023} below  on the mToda Lax operators
\begin{align}
L^*_2(\Lambda-\Lambda^{-1})=(\Lambda-\Lambda^{-1})L_1,\label{constraint1}
\end{align}
then we can get  large BKP hierarchy by
 \begin{align*}
 L^+=L_1,\quad L^-=L_2.
 \end{align*}
 Note that if we denote
  \begin{align*}
 2\mathbf{t}_k=x_k-y_k,\quad2\bar{\mathbf{t}}_k=x_k+y_k,
  \end{align*}
  then only $\pa_{\mathbf{t}_k}=\pa_{x_k}-\pa_{y_k}$ can keep the constraint \eqref{constraint1} (see Section \ref{Section3.1}). So we can put  $\bar{\mathbf{t}}=0$ in the mToda hierarchy corresponding to large BKP hierarchy.  Then we have
  \begin{align*}
\pa_{\mathbf{t}_k}L^{\pm}=[B_{k},L^{\pm}],\quad B_k=(L^+)^k_{\Delta,\geq1}-(L^-)^k_{\Delta^*,\geq1},
\quad(L^+)^*(\Lambda-\Lambda^{-1})=(\Lambda-\Lambda^{-1})L^-,
\end{align*}
which is just the large BKP  hierarchy.

Next if set
\begin{align}
S^+(n, \mathbf{t},\Lambda)=S_1(n, \mathbf{t},-\mathbf{t},\Lambda),\quad
S^-(n, \mathbf{t},\Lambda)=S_2^*(n, \mathbf{t},-\mathbf{t},\Lambda),\label{Wtrans1}
\end{align}
then by  \eqref{Phidefine1} and \eqref{Phidefine2}
\begin{align}
&\Phi_1(n,\mathbf{t},-\mathbf{t},z)=\Psi^+(n,\mathbf{t},z),\quad\Phi^*_1(n,\mathbf{t},-\mathbf{t},z)=\Psi^-(n-1,\mathbf{t},z)
+\Psi^-(n,\mathbf{t},z),
\label{largeBKPandmToda}\\
&\Phi_2(n,\mathbf{t},-\mathbf{t},z)=z^{-1}\Psi^-(n,\mathbf{t},z^{-1}),\quad
\Phi^*_2(n,\mathbf{t},-\mathbf{t},z)=z^{-1}\left(\Psi^+(n-1,\mathbf{t},z^{-1})+
\Psi^+(n,\mathbf{t},z^{-1})\right).\label{largeBKPandmToda2}
\end{align}
Therefore according to \eqref{modifiedTodabilinear}, one can find
\begin{align}
&{\rm Res}_{z}{z^{-1}}\left(\Psi^+(n,\mathbf{t},z)
\Psi^-(n'-1,\mathbf{t}',z)+\Psi^-(n,\mathbf{t},z)
\Psi^+(n'-1,\mathbf{t}',z)\right.\nonumber\\
&\quad\quad\quad\quad\left.+\Psi^+(n,\mathbf{t},z)
\Psi^-(n',\mathbf{t}',z)+\Psi^-(n,\mathbf{t},z)
\Psi^+(n',\mathbf{t}',z)\right)=1,\label{mowavefunctionbilinear}
\end{align}
which is equivalent to the bilinear equation \eqref{wavefunctionbilinear} of  large BKP hierarchy. In fact \eqref{mowavefunctionbilinear} can be derived by using $1+\Lambda^{-1}$ to act on $n''$ in \eqref{wavefunctionbilinear}. Conversely stating from \eqref{mowavefunctionbilinear}, we can obtain
\begin{align*}
S^{+}\Lambda^{-1} S^{-}=\iota_{\Lambda^{-1}}(\Lambda-\Lambda^{-1})^{-1},\quad
 \pa_{\mathbf{t}_k}S^{+}=B_{k}S^{+}-S^{+}\Lambda^{k}, \quad \pa_{\mathbf{t}_k}S^{-}=\Lambda^{k}S^{-}+S^{-}B_{k}^*,
\end{align*}
by the same method in Proposition \ref{Prop:t1kflow}. Then based on this, we can prove \eqref{wavefunctionbilinear} by Proposition \ref{bilinasr}.
Therefore large BKP hierarchy can be regarded as the sub--hierarchy of  mToda hierarchy.

\subsection{Miura links between Toda and mToda hierarchies }\label{twomiura}
In this subsection,  we will discuss  Miura and anti--Miura transformations between Toda and mToda hierarchies. Here the transformation from mToda  to Toda is regarded as  Miura transformation, while the one from Toda to mToda is called  anti--Miura transformation. For more details, one can refer to \cite{Rui2024}.

Firstly for mToda  Lax operators $L_i$ and wave operators $S_i$ $(i=1,2)$, there is a Miura  transformation
\begin{align*}
T_1=c_0^{-1}(n):(L_i,S_i)\rightarrow(\mathcal{L}_i,\mathcal{S}_i),
 \end{align*}
 with $c_0(n)$ being the coefficient of $\Lambda^0$--term in
  $S_1$ $( \ {\rm see} \ \eqref{waveoperatordefine} \ )$ and
\begin{align}
\mathcal{L}_i=T_1L_iT_1^{-1},\quad \mathcal{S}_i=T_1S_i.\label{firsttras}
\end{align}
 Note that by \eqref{mTodalaxdefine}--\eqref{waveoperatordefine}, $\mathcal{L}_i$ and  $\mathcal{S}_i$ have the following  forms
 \begin{align*}
 \mathcal{L}_1=\mathcal{S}_1\Lambda\mathcal{S}_1^{-1}=\Lambda +\sum_{i=0}^{\infty}u_i(n)\Lambda^{-i},\quad
{\mathcal{L}_2}= \mathcal{S}_2\Lambda^{-1} \mathcal{S}_2^{-1}= \overline{u}_{-1}(n)\Lambda^{-1} +\sum_{i=0}^{\infty}  \overline{u}_i(n)\Lambda^{i},
 \end{align*}
and
\begin{align}
\mathcal{S}_1=1+\sum_{k=1}^{\infty}w_k(n)\Lambda^{-k},\quad
{\mathcal{S}_2}=\bar{w}_0(n)+\sum_{k=1}^{\infty}\bar{w}_k(n)\Lambda^{k},\quad \bar{w}_0(n)\neq0,\label{Todawaves}
\end{align}
which are just Toda  Lax and wave operators \cite{takasaki2018jpa}  satisfying
\begin{align}
&\pa_{x_k}\mathcal{L}_i=[(\mathcal{L}_1^k)_{\geq0},\mathcal{L}_i],
\quad \pa_{y_k}{\mathcal{L}_i}=[({\mathcal{L}^k_2})_{<0},{\mathcal{L}_i}],\label{todalax}\\
&\pa_{x_k}\mathcal{S}_1=-(\mathcal{L}^k_1)_{<0}\cdot\mathcal{S}_1,\quad \pa_{y_k}\mathcal{S}_2=-(\mathcal{L}^k_2)_{\geq0}\cdot\mathcal{S}_2,\label{todalax2} \\ &\pa_{y_k}\mathcal{S}_1=(\mathcal{L}^k_2)_{<0}\cdot\mathcal{S}_1,\quad
\pa_{x_k}\mathcal{S}_2=(\mathcal{L}^k_1)_{\geq0}\cdot\mathcal{S}_2.\label{todalax3}
\end{align}
In order to    obtain \eqref{todalax}--\eqref{todalax3} from mToda hierarchy, we need to compare $\Lambda^0$--terms in \eqref{mTodaS1} and \eqref{mTodaS2}, and  find
\begin{align}
&\pa_{x_k}(c_0^{-1})=(\mathcal{L}_1^k)_{\geq0}(c_0^{-1}),\quad
\pa_{y_k}(c_0^{-1})=(\mathcal{L}_2^k)_{<0}(c_0^{-1}),\label{mTodaeig}
\end{align}
where we have used the fact for any operator $A$ and  function $f$ \cite{Cheng12019}
\begin{align*}
(f^{-1}Af)_{\Delta,\geq 1}=f^{-1}A_{\geq 0}f-f^{-1}A_{\geq 0}(f),\quad
(f^{-1}Af)_{\Delta^*,\geq 1}=f^{-1}A_{< 0}f-f^{-1}A_{< 0}(f).
\end{align*}
Further by \eqref{mTodalaxeq}, \eqref{mTodaS1} and \eqref{mTodaS2}, we can get \eqref{todalax}--\eqref{todalax3}.
 There is also a  self--transformation for Toda hierarchy,
\begin{align*}
\pi_0: (\mathcal{L}_i,\mathcal{S}_i,x_k,y_k)\rightarrow
(\tilde{\mathcal{L}}_i,\tilde{\mathcal{S}}_i,\tilde{x}_k,\tilde{y}_k),
\end{align*}
with
$(\tilde{\mathcal{L}}_i,\tilde{\mathcal{S}}_i,\tilde{x}_k,\tilde{y}_k)=
(\bar{w}_0\mathcal{L}^*_{3-i}\bar{w}^{-1}_0,
\bar{w}_0(\mathcal{S}^{-1}_{3-i})^*,-y_k,-x_k)$ satisfying Toda  hierarchy, that is,
\begin{align*}
\pa_{\tilde{x}_k}{\tilde{\mathcal{L}}_i}=[({\tilde{\mathcal{L}}^k_1})_{\geq0},\tilde{L}_i],
\quad \pa_{\tilde{y}_k}{\tilde{\mathcal{L}}_i}
=[({\tilde{\mathcal{L}}^k_2})_{<0},\tilde{\mathcal{L}}_i].
\end{align*}
Here we have used
 \begin{align}
\pa_{x_k}\bar{w}_0\cdot\bar{w}^{-1}_0=(\mathcal{L}_1^k)_{[0]},\quad
\pa_{y_k}\bar{w}_0\cdot\bar{w}^{-1}_0=-(\mathcal{L}_2^k)_{[0]},\label{todafirst}
 \end{align}
which are derived from \eqref{todalax2} and \eqref{todalax3}. Note that different from the self--transformation $\pi_1$ in mToda hierarchy, we can found $\pi_0^2=1$.

If introduce  Toda  wave function $\Psi_i$ and adjoint  wave function $\Psi^*_i$ as follows
\begin{align*}
&{\Psi}_1(n,x,y,z)=\mathcal{S}_1(n,x,y,\Lambda)e^{\xi(x,\Lambda)}(z^n)
,\quad{\Psi}_2(n,x,y,z)=\mathcal{S}_2(n,x,y,\Lambda)e^{\xi(y,\Lambda^{-1})}(z^{n}),\\ &{\Psi}^*_1(n,x,y,z)=\Big({\mathcal{S}_1}(n,x,y,\Lambda)^{-1}\Big)^*
e^{-\xi(x,\Lambda^{-1})}(z^{-n}),\quad
{\Psi}^*_2(n,x,y,z)=\Big({\mathcal{S}_2}(n,x,y,\Lambda)^{-1}\Big)^*
e^{-\xi(y,\Lambda)}(z^{-n}),
\end{align*}
then by \eqref{todalax2}  and \eqref{todalax3}, we can get
\begin{align*}
\mathcal{L}_1(\Psi_1)=z\Psi_1,\quad\mathcal{L}_2(\Psi_2)=z^{-1}\Psi_2,\quad
\pa_{x_k}\Psi_i=(\mathcal{L}^k_1)_{\geq0}(\Psi_i),\quad
\pa_{y_k}\Psi_i=(\mathcal{L}^k_2)_{<0}(\Psi_i),\\
\mathcal{L}^*_1(\Psi^*_1)=z\Psi^*_1,\quad\mathcal{L}^*_2(\Psi^*_2)=z^{-1}\Psi^*_2,\quad
\pa_{x_k}\Psi^*_i=-((\mathcal{L}^k_1)_{\geq0})^*(\Psi^*_i),\quad
\pa_{y_k}\Psi^*_i=-((\mathcal{L}^k_2)_{<0})^*(\Psi^*_i).
\end{align*}
Here wave function $\Psi_i$ and adjoint  wave function $\Psi^*_i$ are special cases of Toda  eigenfunction $q(n)$ and adjoint eigenfunction $r(n)$ respectively, which are defined by
\begin{align*}
&\pa_{x_k}q(n)=(\mathcal{L}_1(n)^k)_{\geq0}(q(n)),\quad
\pa_{{y}_k}q(n)=(\mathcal{L}_2(n)^k)_{<0}(q(n)),\\
&\pa_{x_k}r(n)=-(\mathcal{L}_1(n)^k)^*_{\geq0}(r(n)),\quad
\pa_{y_k}r(n)=-(\mathcal{L}_2(n)^k)^*_{<0}(r(n)).
\end{align*}
Note that by \eqref{mTodaeig}, $c_0^{-1}$ is the Toda   eigenfunction  with respect to $\mathcal{L}_i$ in \eqref{firsttras}.

Then by \eqref{Phidefine1}, \eqref{Phidefine2} and \eqref{firsttras},
\begin{align}
&{\Psi}_1(n,x,y,z)=c_0^{-1}(n){\Phi}_1(n,x,y,z),
\quad{\Psi_1}^*(n,x,y,z)=c_0(n)\Big({\Phi}^*_1(n,x,y,z)-{\Phi}^*_1(n+1,x,y,z)\Big),\label{Todawavefun}
\\
&{\Psi}_2(n,x,y,z)=c_0^{-1}(n){\Phi}_2(n,x,y,z),\quad
{\Psi_2}^*(n,x,y,z)=c_0(n)\Big({\Phi}^*_2(n+1,x,y,z)-{\Phi}^*_2(n,x,y,z)\Big).\label{Todawavefun2}
\end{align}
So  by using $1-\Lambda$ to act on $n'$ in \eqref{modifiedTodabilinear} and multiplying $c_0^{-1}(n)c_0(n')$, we can get
\begin{align}\label{wavebilinearcomponent}
    \oint_{z=\infty}\frac{dz}{2\pi iz}\Psi_1(n;x,y;z)\Psi_1^*(n';x',y';z)
    =\oint_{z=0}\frac{dz}{2\pi iz}\Psi_2(n;x,y;z)\Psi_2^*(n';x',y';z),
\end{align}
which is just  Toda bilinear equation.
On this basis, it can be proved that there exists a  tau function $\tau^{{\rm Toda}}_n(x,y)$   \cite{takasaki2018jpa} such that
\begin{eqnarray}
\Psi_1(n,x,y,z)=\frac{\tau^{{\rm Toda}}_n
(x-[z^{-1}],y)}
{\tau^{{\rm Toda}}_n(x,y)}e^{\xi(x,z)}z^n,&
&{\Psi_2}(n,x,y,z)=\frac{\tau^{{\rm Toda}}_{n+1}
(x,y-[z])}
{\tau^{{\rm Toda}}_n(x,y)}
e^{\xi(y,z^{-1})}z^{n},\label{TLtau}\\
{\Psi_1}^*(n,x,y,z)=\frac{\tau^{{\rm Toda}}_{n+1}
(x+[z^{-1}],y)}
{\tau^{{\rm Toda}}_{n+1}(x,y)}e^{-\xi(x,z)}z^{-n},&
&{\Psi_2}^*(n,x,y,z)=\frac{\tau^{{\rm Toda}}_{n}
(x,y
+[z])}{\tau^{{\rm Toda}}_{n+1}(x,y)}
e^{-\xi(y,z^{-1})}z^{-n}.\label{TLtau2}
\end{eqnarray}

From \eqref{Todawavefun} and \eqref{Todawavefun2}, we can find by $c_0(n)=\frac{\tau_{0,n}(x,y)}{\tau_{1,n}(x,y)}$, \eqref{mTLtau} and \eqref{TLtau} that
$$\frac{\tau_n^{\rm Toda}(x-[z^{-1}],y)}{\tau_n^{\rm Toda}(x,y)}=\frac{\tau_{0,n}(x-[z^{-1}],y)}{\tau_{0,n}(x,y)},\quad \frac{\tau_{n+1}^{\rm Toda}(x,y-[z])}{\tau_n^{\rm Toda}(x,y)}=\frac{\tau_{0,n+1}(x,y-[z])}{\tau_{0,n}(x,y)},$$
which implies
$$\big(e^{-\xi(\tilde{\pa}_x,z^{-1})}-1\big)\big({\rm log}\tau_n^{\rm Toda}(x,y)-{\rm log}\tau_{0,n}(x,y)\big)=0,\quad \big(\Lambda e^{-\xi(\tilde{\pa}_y,z)}-1\big)\big({\rm log}\tau_n^{\rm Toda}(x,y)-{\rm log}\tau_{0,n}(x,y)\big)=0.$$
Note that  if function $h(n,x,y)$ satisfy $\big( e^{-\xi(\tilde{\pa}_x,z^{-1})}-1\big)h(n,x,y)=0$ and $\big( \Lambda e^{-\xi(\tilde{\pa}_y,z)}-1\big)h(n,x,y)=0$, then $h(n,x,y)$ is some constant without depending on $n$, $x$ and $y$.
Therefore $$\tau_{0,n}(x,y)={\rm const}\cdot\tau_n^{\rm Toda}(x,y).$$

Besides the first Miura transformation $T_1$, there also
  exists another  Miura transformation
\begin{align*}
T_2: (L_i,S_i)\xrightarrow {\pi_1^{-1}} (\widetilde{L}_i,\widetilde{S}_i)\xrightarrow {{T_1}} (\widetilde{\mathcal{L}}_i,\widetilde{\mathcal{S}}_i)\xrightarrow {\pi_0}({\mathcal{L}}_i,{\mathcal{S}}_i),
\end{align*}
that is ${T_2=\pi_0\circ T_1\circ \pi_1^{-1}}:$ $(L_i,S_i)\rightarrow(\mathcal{L}_i,\mathcal{S}_i)$. It can be found that
 $${T}_2=c_0^{-1}(n+1)\Delta,$$
and
\begin{align}
\mathcal{L}_1=T_2\cdot L_1\cdot \iota_{\Lambda^{-1}}T_2^{-1},\quad
\mathcal{L}_2=T_2\cdot L_2\cdot \iota_{\Lambda}T_2^{-1},\quad
\mathcal{S}_1={T}_2\cdot {S}_1\cdot\Lambda^{-1},\quad
\mathcal{S}_2=-{T}_2\cdot {S}_2,\label{Todadefinetra}
\end{align}
satisfying \eqref{todalax}--\eqref{todalax3}. Further by \eqref{Phidefine1}, \eqref{Phidefine2} and \eqref{Todadefinetra}
\begin{align}
&{\Psi}_1(n,x,y,z)=z^{-1}c_0^{-1}(n+1)\Big({\Phi}_1(n+1,x,y,z)-{\Phi}_1(n,x,y,z)\Big),\quad {\Psi_1}^*(n,x,y,z)=zc_0(n+1){\Phi}^*_1(n+1,x,y,z),\label{mTodatauToda2}\\
&{\Psi}_2(n,x,y,z)=-c_0^{-1}(n+1)\Big({\Phi}_2(n+1,x,y,z)-{\Phi}_2(n,x,y,z)\Big),\quad
{\Psi_2}^*(n,x,y,z)=c_0(n+1){\Phi}^*_2(n+1,x,y,z),\label{mTodatauToda3}
\end{align}
satisfying Toda bilinear equation \eqref{wavebilinearcomponent}.
In particularly  by comparing the $\Lambda^0$--terms of $\pa_{x_k}S_2$ and $\pa_{y_k}S_2$,
 one can derive
\begin{align}
\pa_{x_k}c_0(n+1)=-((\mathcal{L}_1(n)^k)_{\geq0})^*\left(c_0(n+1)\right),\quad
\pa_{y_k}c_0(n+1)=-((\mathcal{L}_2(n)^k)_{<0})^*\left(c_0(n+1)\right),\label{adjointeg}
\end{align}
which shows   $c_0(n+1)$ is the Toda  adjoint eigenfunction  respect to $\mathcal{L}_i$ in \eqref{Todadefinetra}.
Similarly from \eqref{mTodatauToda2} and \eqref{mTodatauToda3}, we can find  by \eqref{mTLtau2} and \eqref{TLtau2}
$$\tau_{1,n}(x,y)={\rm const}\cdot\tau_n^{\rm Toda}(x,y).$$

Let us summarize above results in the following proposition.
\begin{proposition}\label{muratras}
Given mToda objects:  Lax operators $L_i$, wave operators ${S}_i$, wave functions $\Phi_i$ and adjoint wave functions $\Phi_i^*$ $(i=1,2)$, if set $\mathcal{L}_i$,  $\mathcal{S}_i$,  $\Psi_i$ and  $\Psi_i^*$ in the form of Table I
\begin{center}
\begin{tabular}{lll}
\multicolumn{3}{c}{Table I. Miura transformations: mToda $\rightarrow$ Toda}\\
\hline \hline
 &$T_1=c_0^{-1}$ &\quad$T_2=c_0^{-1}(n+1)\Delta$ \ \\
\hline
{\rm Wave operator} &$\mathcal{S}_1={T}_1\cdot{S}_1$& \quad$\mathcal{S}_1={T}_2\cdot{S}_1\Lambda^{-1}$\\
&$\mathcal{S}_2={T}_1\cdot{S}_2$& \quad$\mathcal{S}_2=-{T}_2\cdot{S}_2$\\
{\rm Lax operator} &$\mathcal{L}_1=
{T}_1\cdot{L}_1\cdot{T}_1^{-1}$ &\quad $ \mathcal{L}_1=T_2\cdot L_1\cdot(\iota_{\Lambda^{-1}}T_2^{-1})$\\
  &$\mathcal{L}_2=
{T}_1\cdot{L}_2\cdot{T}_1^{-1}$ &\quad $ \mathcal{L}_2=T_2\cdot L_2\cdot(\iota_{\Lambda}T_2^{-1})$\\
{\rm Wave function}&${\Psi}_1=T_1({\Phi}_1)$& \quad$ {\Psi}_1=z^{-1}T_2({\Phi}_1) $ \\
&${\Psi}_2=T_1({\Phi}_2)$& \quad ${\Psi}_2=-T_2({\Phi}_2)$ \\
{\rm Adjoint wave function}&${\Psi_1}^*=-{T}_1^{-1}\Delta({\Phi}^*_1)$&\quad ${\Psi_1}^*=zc_0(n+1)\Lambda({\Phi}^*_1)$ \\
& ${\Psi_2}^*={T}_1^{-1}\Delta({\Phi}^*_2)$ & \quad
${\Psi_2}^*=c_0(n+1)\Lambda({\Phi}^*_2)$ \\
{\rm Tau function} &$\tau_{n}^{\rm Toda}={\rm const}\cdot\tau_{0,n}$& \quad$\tau_{n}^{\rm Toda}={\rm const}\cdot\tau_{1,n}$\\
\hline
\end{tabular}
\end{center}
 then $\Psi_i$ and $\Psi_i^*$ satisfy Toda bilinear equation \eqref{wavebilinearcomponent}, and
  $(\mathcal{L}_i, \mathcal{S}_i)$ satisfies Toda  evolution equations  \eqref{todalax}--\eqref{todalax3}.
\end{proposition}
\begin{remark}\label{taufunctionmToda}
Given mToda hierarchy having tau functions $\big(\tau_{0,n},\tau_{1,n}\big)$, $\tau_{0,n}$ and $\tau_{1,n}$ will be Toda two tau functions by Proposition \ref{muratras}. And if
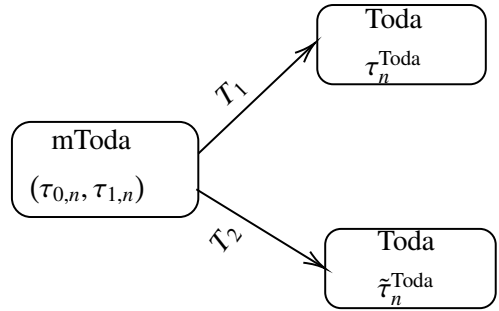
\begin{center}

\tikzset{every picture/.style={line width=0.75pt}} 

\begin{tikzpicture}[x=0.75pt,y=0.75pt,yscale=-1,xscale=1]

\draw   (102,120.6) .. controls (102,115.3) and (106.3,111) .. (111.6,111) -- (186.4,111) .. controls (191.7,111) and (196,115.3) .. (196,120.6) -- (196,149.4) .. controls (196,154.7) and (191.7,159) .. (186.4,159) -- (111.6,159) .. controls (106.3,159) and (102,154.7) .. (102,149.4) -- cycle ;
\draw   (256,60) .. controls (256,55.58) and (259.58,52) .. (264,52) -- (333,52) .. controls (337.42,52) and (341,55.58) .. (341,60) -- (341,84) .. controls (341,88.42) and (337.42,92) .. (333,92) -- (264,92) .. controls (259.58,92) and (256,88.42) .. (256,84) -- cycle ;
\draw   (260,173) .. controls (260,168.58) and (263.58,165) .. (268,165) -- (338,165) .. controls (342.42,165) and (346,168.58) .. (346,173) -- (346,197) .. controls (346,201.42) and (342.42,205) .. (338,205) -- (268,205) .. controls (263.58,205) and (260,201.42) .. (260,197) -- cycle ;
\draw    (196,127) -- (253.55,72.38) ;
\draw [shift={(255,71)}, rotate = 136.49] [color={rgb, 255:red, 0; green, 0; blue, 0 }  ][line width=0.75]    (10.93,-3.29) .. controls (6.95,-1.4) and (3.31,-0.3) .. (0,0) .. controls (3.31,0.3) and (6.95,1.4) .. (10.93,3.29)   ;
\draw    (195,145) -- (258.3,183.95) ;
\draw [shift={(260,185)}, rotate = 211.61] [color={rgb, 255:red, 0; green, 0; blue, 0 }  ][line width=0.75]    (10.93,-3.29) .. controls (6.95,-1.4) and (3.31,-0.3) .. (0,0) .. controls (3.31,0.3) and (6.95,1.4) .. (10.93,3.29)   ;

\draw (149.5,133.5) node   [align=left] {\begin{minipage}[lt]{57.8pt}\setlength\topsep{0pt}
 \ \ \ {\rm mToda}\\ \ \ \ \   $\big({\tau}_{0,n},{\tau}_{1,n}\big)$
\end{minipage}};
\draw (264,52) node [anchor=north west][inner sep=0.75pt]   [align=left] { \ \ \ \ {\rm Toda} \ \ \ \\ \ \ \ \  ${\tau}^{\rm Toda}_n$};
\draw (274,164) node [anchor=north west][inner sep=0.75pt]   [align=left] { \ \ \ {\rm Toda} \ \ \\ \ \ \ $\tilde{\tau}^{\rm Toda}_n$};
\draw (201.2,96.56) node [anchor=north west][inner sep=0.75pt]  [rotate=-316.6,xslant=-0.02] [align=left] {$T_1$};
\draw (198.63,172.38) node [anchor=north west][inner sep=0.75pt]  [rotate=-306.24] [align=left] {$T_2$};

\end{tikzpicture}

\end{center}
then $\big(\tau_{n}^{\rm Toda},\tilde{\tau}_{n}^{\rm Toda}\big)$ will satisfy bilinear equation of mToda hierarchy \cite{Rui2024}, that is
\begin{align*}
&{\rm Res}_{z}\left(z^{n-n'-1}\tau_{n}^{\rm Toda}(x-[z^{-1}],y)
\cdot\tilde{\tau}_{n'}^{\rm Toda}(x'+{[z^{-1}]},y)e^{\xi(x-x',z)}\right.\nonumber\\
&\quad\quad\quad\left.
+z^{n'-n-2}\tau_{n+1}^{\rm Toda}(x,y-{[z^{-1}]})
\cdot\tilde{\tau}_{n'-1}^{\rm Toda}(x',y'+[z^{-1}])e^{\xi(y-y',z)}\right)=\tilde{\tau}_{n}^{\rm Toda}
(x,y)\tau_{n'}^{\rm Toda}(x',y').
\end{align*}
\end{remark}
Conversely from Toda to mToda, there are also two kinds of anti--Miura transformations  in the following proposition.
\begin{proposition}\label{antimiuraToda}\cite{Rui2024}
Given Toda objects: Lax operators $\mathcal{L}_i$, wave operators $\mathcal{S}_i$, wave functions $\Psi_i$, adjoint wave functions $\Psi_i^*$ $(i=1,2)$,   eigenfunction $q(n)$ and adjoint eigenfunction $r(n)$, if set $L_i$, $S_i$, $\Phi_i$ and $\Phi_i^*$ in the form of  Table II
\begin{center}
\begin{tabular}{lll}
\multicolumn{3}{c}{Table II. Anti--Miura transformations: Toda $\rightarrow$ mToda}\\
\hline \hline
&$\mathcal{T}_1=q(n)^{-1}$ &$\mathcal{T}_2=\Delta^{-1}r(n)$   \\
\hline
{\rm Wave operator} &${S}_1=\mathcal{T}_1\cdot\mathcal{S}_1$& ${S}_1=(\iota_{\Lambda^{-1}}\mathcal{T}_2)\cdot\mathcal{S}_1\cdot\Lambda$ \\
&${S}_2=\mathcal{T}_1\cdot\mathcal{S}_2$&
${S}_2=-(\iota_{\Lambda}\mathcal{T}_2)\cdot\mathcal{S}_2$ \\
{\rm Lax operator} &${L}_1=
\mathcal{T}_1\cdot\mathcal{L}_1\cdot\mathcal{T}_1^{-1}$& ${L}_1=(\iota_{\Lambda^{-1}}\mathcal{T}_2)\cdot \mathcal{L}_1\cdot\mathcal{T}_2^{-1}$\\
 &${L}_2=
\mathcal{T}_1\cdot\mathcal{L}_2\cdot\mathcal{T}_1^{-1}$& ${L}_2=(\iota_{\Lambda}\mathcal{T}_2)\cdot \mathcal{L}_2\cdot\mathcal{T}_2^{-1}$
 \\
{\rm Wave function}&$\Phi_1=\mathcal{T}_1(\Psi_1)$& $ {\Phi}_1=z(\iota_{\Lambda^{-1}}\mathcal{T}_2)(\Psi_1) $\\
&$\Phi_2=\mathcal{T}_1(\Psi_2)$& $\Phi_2=-(\iota_{\Lambda}\mathcal{T}_2)({\Psi_2})$ \\
{\rm Adjoint wave function}&$\Phi_1^*=-(\iota_{\Lambda}\Delta^{-1})\mathcal{T}_1^{-1}(\Psi^*_1)$ & $\Phi_1^*=\Lambda^{-1}r^{-1}(n)z^{-1}(\Psi_1^*)$\\
& $\Phi_2^*=(\iota_{\Lambda^{-1}}\Delta^{-1})\mathcal{T}_1^{-1}({\Psi^*_2})$&
$\Phi_2^*=\Lambda^{-1}r^{-1}(n)(\Psi_2^*)$ \\
{\rm Tau function} &$\tau_{0,n}={\rm const}\cdot\tau_{n}^{\rm Toda}$& $\tau_{0,n}={\rm const}\cdot r(n-1)\tau_{n}^{\rm Toda}$\\
 &$\tau_{1,n}={\rm const}\cdot q(n)\tau_{n}^{\rm Toda}$& $\tau_{1,n}={\rm const}\cdot \tau_{n}^{\rm Toda}$\\
\hline
\end{tabular}
\end{center}
then $\Phi_i$ and $\Phi_i^*$ satisfy \eqref{modifiedTodabilinear}, and  $(L_i ,S_i)$ satisfies mToda  evolution equations  \eqref{mTodalaxeq}, \eqref{mTodaS1} and \eqref{mTodaS2}.
\end{proposition}
\begin{remark}
Similar to Miura transformations in Proposition \ref{muratras}, the  transformation $\mathcal{T}_2$ can also be seen from the commuting diagram below
\begin{align*}
({\mathcal{L}}_i,{\mathcal{S}}_i)\xrightarrow {\pi_0}(\widetilde{\mathcal{L}}_i,\widetilde{\mathcal{S}}_i)\xrightarrow {{\mathcal{T}_1}} (\widetilde{L}_i,\widetilde{S}_i)\xrightarrow {\pi_1^{-1}}(L_i,S_i),
\end{align*}
that is ${\mathcal{T}_2=\pi_1^{-1}\circ \mathcal{T}_1\circ \pi_0}$. Here we have used the following facts, if $q(n)$ and $r(n)$ are the Toda  eigenfunction and adjoint eigenfunction corresponding to Lax operator $\mathcal{L}_i$ respectively, then $\bar{w}_0(n)^{-1}q(n)$ and $\bar{w}_0(n)r(n)$ can be regarded as the adjoint eigenfunction and eigenfunction corresponding to Lax operator $\widetilde{\mathcal{L}}_i=\bar{w}_0\mathcal{L}^*_{3-i}\bar{w}^{-1}_0$.
  \end{remark}

\subsection{The relations of tau functions between  large BKP   and  Toda hierarchies}
Since  large BKP hierarchy can be regarded as the sub--hierarchy of mToda hierarchy, here in this subsection, we restrict  Miura links to the case of large BKP hierarchy and show the relations of tau functions with Toda hierarchy.

Starting from large BKP  wave operators $S^+=\sum_{j=0}^{+\infty} b_j(n)\Lambda^{-j}$ and  $S^-=\sum_{j=0}^{+\infty} \Lambda^{-j}\bar{b}_j(n)$,
 \begin{align*}
S^{+}\Lambda^{-1} S^{-}=\iota_{\Lambda^{-1}}(\Lambda-\Lambda^{-1})^{-1},\quad
 \pa_{\mathbf{t}_k}S^{+}=B_{k}S^{+}-S^{+}\Lambda^{k}, \quad \pa_{\mathbf{t}_k}S^{-}=\Lambda^{k}S^{-}+S^{-}B_{k}^*,
\end{align*}
where
$
B_{k}=\Big(S^+\Lambda^{k} (S^+)^{-1}\Big)_{\Delta,\geq1}-\Big((S^-)^*\Lambda^{-k} ((S^-)^{-1})^*\Big)_{\Delta^*,\geq1},
$
if denote
 \begin{align*}
\mathcal{S}_1=T_1S^{+},\quad\mathcal{S}_2=T_1(S^{-})^*,\quad T_1=b_0(n)^{-1},
 \end{align*}
 where $b_0(n)$ is the coefficient of $\Lambda^0$--term in
  $S^+$,
we can obtain by $\pa_{\mathbf{t}_k}=\pa_{{x}_k}-\pa_{{y}_k}$  and Proposition \ref{muratras}
\begin{align*}
\mathcal{S}_1\Lambda^{-1}\mathcal{S}_2^*
=T_1\Big(\iota_{\Lambda^{-1}}(\Lambda-\Lambda^{-1})\Big)^{-1}T_1^*,\quad
\pa_{\mathbf{t}_k}\mathcal{S}_1
=-\Big((\mathcal{L}_1^k)_{<0}+(\mathcal{L}_2^k)_{<0}\Big)\mathcal{S}_1,\quad
\pa_{\mathbf{t}_k}\mathcal{S}_2
=\Big((\mathcal{L}_1^k)_{\geq0}+(\mathcal{L}_2^k)_{\geq0}\Big)\mathcal{S}_2.
\end{align*}
If set $\bar{w}_0(n)=e^{\varphi(n)}$ ( $\Lambda^0$--term in $\mathcal{S}_2$, see \eqref{Todawaves} ), we can obtain the {\bf balance condition I} \cite{Krichever2023}
\begin{align}
e^{\varphi(n)}=b_0(n)^{-1}b_0(n+1)^{-1},\label{balancecond}
\end{align}
which can be  derived  by comparing coefficients of $\Lambda^0$--terms in $\mathcal{S}_2=T_1(S^{-})^*$ and $\Lambda^{-1}$--terms in $S^{+}\Lambda^{-1} S^{-}=\iota_{\Lambda^{-1}}(\Lambda-\Lambda^{-1})^{-1}$.

 While for  Lax equation of large BKP hierarchy
 \begin{align*}
 \pa_{\mathbf{t}_k}L^{\pm}=[B_{k},L^{\pm}],\quad
 (L^+)^*(\Lambda-\Lambda^{-1})=(\Lambda-\Lambda^{-1})L^-,
\end{align*}
where $L^{\pm}=\sum_{j=-1}^\infty u^{\pm}_j\Lambda^{\mp j}$ and $B_k=(L^+)^k_{\Delta,\geq1}-(L^-)^k_{\Delta^*,\geq1}$,
if set
\begin{align*}
\mathcal{L}_1=T_1L^{+}T_1^{-1},\quad \mathcal{L}_2=T_1L^{-}T_1^{-1},\quad T_1=b_0(n)^{-1},
\end{align*}
then by Proposition \ref{muratras} and balance condition I \eqref{balancecond}, we can get
\begin{align}
&\pa_{\mathbf{t}_k}\mathcal{L}_i
=[\mathcal{B}_k,\mathcal{L}_i],\quad \mathcal{B}_k=(\mathcal{L}_1^k)_{\geq0}-(\mathcal{L}_2^k)_{<0},\label{muiralax1}\\
&{ \mathcal{L}_1}^*(\Gamma_1-\Gamma_1^*)
=(\Gamma_1-\Gamma_1^*){ \mathcal{L}_2},\quad\Gamma_1=e^{-\varphi(n)}\Lambda,\label{muiralax2}
\end{align}
which can be  regarded as the sub--Toda hierarchy corresponding to  large BKP hierarchy.

Note that balance condition I \eqref{balancecond} is very important (see \cite{Krichever2023,Prokofev2023} ). Next we will show that the balance condition I \eqref{balancecond} is invariant under the flow $\pa_{\mathbf{t}_k}$. To do this,  we need the lemma below.
\begin{lemma}\label{operatorequvalent}
Given two  operators $A=\sum_{j=-k}^{+\infty}a_j(n)\Lambda^{-j}$ and $B=\sum_{j=-k}^{+\infty}b_j(n)\Lambda^{j}$ $(k\geq1)$ satisfying
\begin{align*}
A^{*}\cdot(\Lambda-\Lambda^{-1})=(\Lambda-\Lambda^{-1})\cdot B,
\end{align*}
we have
\begin{align}
A_{[0]}+B_{[0]}=
(\Lambda+1)\left(A_{\geq0}-B_{\leq-1}\right){|}_{\Lambda=1}.\label{operatorarelation}
\end{align}
\end{lemma}
\begin{proof}
Firstly if set $A=\tilde{A}(\Lambda-\Lambda^{-1})$ and $B=\tilde{B}(\Lambda^{-1}-\Lambda)$,
then $\tilde{A}^*=\tilde{B}$. So if denote $\tilde{A}=\sum_{i}\tilde{A}_i(n)\Lambda^i$, then
\begin{align*}
A_{[0]}=\tilde{A}_{-1}(n)-\tilde{A}_1(n),\quad
B_{[0]}=\tilde{A}_{-1}(n+1)-\tilde{A}_1(n-1).
\end{align*}
While note that
\begin{align*}
A_{\geq0}=\sum_{i\geq-1}\tilde{A}_i(n)
\Lambda^{i+1}-\sum_{i\geq1}\tilde{A}_i(n)
\Lambda^{i-1},\quad
B_{\leq-1}
=\sum_{i\geq0}\tilde{A}_i(n-i)
\Lambda^{-i-1}-\sum_{i\geq2}\tilde{A}_i(n-i)
\Lambda^{-i+1},
\end{align*}
therefore
$\left(A_{\geq0}-B_{\leq-1}\right){|}_{\Lambda=1}
=\tilde{A}_{-1}(n)-\tilde{A}_1(n-1).$
So finally we can get \eqref{operatorarelation}.
\end{proof}
By \eqref{mTodaeig}, \eqref{todafirst} and Lemma \ref{operatorequvalent}, we can get the following proposition.
\begin{proposition}\label{balancecondition1}
The balance condition I \eqref{balancecond} is invariant under the flows $\pa_{\mathbf{t}_k}$,
$$\pa_{\mathbf{t}_k}\Big(e^{\varphi(n)}-b_0(n)^{-1}b_0(n+1)^{-1}\Big)=0.$$
\end{proposition}
\begin{remark}\label{miura1remark}
In \cite{Krichever2023}, large BKP hierarchy is obtained  from Toda hierarchy.  In fact given Toda Lax operators
 $$\mathcal{L}_1=\Lambda +\sum_{i=0}^{\infty}u_i(n)\Lambda^{-i},\quad{\mathcal{L}_2}= \overline{u}_{-1}(n)\Lambda^{-1} +\sum_{i=0}^{\infty}  \overline{u}_i(n)\Lambda^{i},$$
 with $\overline{u}_{-1}(n)=e^{\varphi(n)-\varphi(n-1)}$, if denote
  \begin{align}
 L^+=q^{-1}\mathcal{L}_1q=\sum_{j=-1}^\infty u^+_j\Lambda^{-j},\quad
L^-=q^{-1}\mathcal{L}_2q=\sum_{j=-1}^\infty u^-_j\Lambda^{j},\label{largeBKPlaxop}
  \end{align}
  and impose large BKP constraint  $
  (L^{+})^*{\left(\Lambda-\Lambda^{-1}\right)}=(\Lambda-\Lambda^{-1}) L^{-},$
then  we can get  $\pa_{\mathbf{t}_k}L^{\pm}=[B_{k},L^{\pm}],$
   where $B_k$ is given by
   \begin{align*}
B_k&=
q^{-1}( \mathcal{L}_1^k)_{\geq0}q-q^{-1}( \mathcal{L}_2^k)_{<0}q-\pa_{\mathbf{t}_k}{\rm log}q
=({L}^+)^k _{\geq0}-({L}^-)^k _{<0}-\pa_{\mathbf{t}_k}{\rm log}q.
\end{align*}
Note that the relation of $B_k$ with $L^{\pm}$ is not clear. In fact from large BKP constraint, we can know $u^+_{-1}=u^-_{-1}$ implying  $$e^{\varphi(n)}=q(n)q(n+1),$$ that is the balance condition I. Further by $\pa_{\mathbf{t}_k}\varphi(n)=(\mathcal{L}_1^k)_{[0]}+(\mathcal{L}_2^k)_{[0]}
=(L^+)^k_{[0]}+(L^-)^k_{[0]}$ {\rm(see \eqref{todafirst})} and  Lemma \ref{operatorequvalent}, we can find  \begin{align*}
\pa_{\mathbf{t}_k}{\rm log}q=q^{-1}\left((\mathcal{L}_1^k) _{\geq0}-(\mathcal{L}_2^k) _{<0}\right)(q)
=\left(({L}^+)^k _{\geq0}-({L}^-)^k _{<0}\right)\big{|}_{\Lambda=1}.
\end{align*}
Therefore $\pa_{\mathbf{t}_k}q=\left((\mathcal{L}_1^k) _{\geq0}-(\mathcal{L}_2^k) _{<0}\right)(q)$ and $B_{k}=(L^{+})^k_{\geq 0}-(L^{-})^k_{\leq -1}-\left((L^{+})^k_{\geq 0}-(L^{-})^k_{\leq -1}\right)\big{|}_{\Lambda=1}$, which is just \eqref{b1kandlax}. Further by Remark \ref{BkdefineDelta}, we have $B_k=(L^+)^k_{\Delta,\geq1}-(L^-)^k_{\Delta^*,\geq1}$.
So this special transformation \eqref{largeBKPlaxop} is just    anti--Miura transformation $\mathcal{T}_1=q^{-1}$.
\end{remark}
 By \eqref{largeBKPandmToda}, \eqref{largeBKPandmToda2} and \eqref{Todawavefun}--\eqref{wavebilinearcomponent}, we can get
 \begin{align*}
&\Psi_1(n,\mathbf{t},-\mathbf{t},z)=b_0(n)^{-1}\Psi^+(n,\mathbf{t},z),
\quad\Psi_2(n,\mathbf{t},-\mathbf{t},z)=b_0(n)^{-1}z^{-1}\Psi^-(n,\mathbf{t},z^{-1}),
\\
&\Psi_1(n,\mathbf{t},-\mathbf{t},z)^*=b_0(n)
\Big(\Psi^-(n-1,\mathbf{t},z)-\Psi^-(n+1,\mathbf{t},z)\Big),\\
&\Psi_2(n,\mathbf{t},-\mathbf{t},z)^*
=b_0(n)z^{-1}
\Big(\Psi^+(n+1,\mathbf{t},z^{-1})-\Psi^+(n-1,\mathbf{t},z^{-1})\Big),
\end{align*}
satisfying
\begin{align}
    \oint_{z=\infty}\frac{dz}{2\pi iz}\Psi_1(n;\mathbf{t},-\mathbf{t};z)\Psi_1^*(n';\mathbf{t}',-\mathbf{t}';z)
    =\oint_{z=0}\frac{dz}{2\pi iz}\Psi_2(n;\mathbf{t},-\mathbf{t};z)\Psi_2^*(n';\mathbf{t}',-\mathbf{t}';z),\label{TodaMiura}
\end{align}
which is just the Toda bilinear equation.

Similarly by Proposition \ref{muratras}, there also exists another  Miura transformation ${T}_2=b_0(n+1)^{-1}\Delta$ such that
\begin{align}
   \mathcal{L}_1=T_2\cdot L^+\cdot \iota_{\Lambda^{-1}}T_2^{-1},\quad
\mathcal{L}_2=T_2\cdot L^-\cdot \iota_{\Lambda}T_2^{-1},\quad
\mathcal{S}_1={T}_2\cdot {S}^+\Lambda^{-1},\quad
\mathcal{S}_2=-{T}_2\cdot ({S}^-)^*,\label{laegeBKPmiura2}
\end{align}
satisfying
\begin{align}
\mathcal{S}_1\mathcal{S}_2^*
=-T_2\cdot\Big(\iota_{\Lambda^{-1}}(\Lambda-\Lambda^{-1})^{-1}\Big)\cdot T_2^*,\quad
T_2\cdot\iota_{\Lambda}(\Lambda-\Lambda^{-1})^{-1}\cdot T_2^*\cdot{ \mathcal{L}_1}^*
= { \mathcal{L}_2}\cdot T_2\cdot\iota_{\Lambda}(\Lambda-\Lambda^{-1})^{-1}\cdot T_2^*.\label{Todalaxmiura2}
\end{align}
According to  \eqref{laegeBKPmiura2} and  comparing  $\Lambda^{-1}$--terms in $S^{+}\Lambda^{-1} S^{-}=\iota_{\Lambda^{-1}}(\Lambda-\Lambda^{-1})^{-1}$, we can obtain the {\bf balance condition II}
\begin{align}
e^{\varphi(n)}=b_0(n+1)^{-2}.\label{condition2}
\end{align}

\begin{lemma}\label{operatorequvalentDelta}
Given two  operators $A=\sum_{j=-k}^{+\infty}a_j(n)\Lambda^{-j}$ and $B=\sum_{j=-k}^{+\infty}b_j(n)\Lambda^{j}$ $(k\geq1)$ satisfying
\begin{align*}
A^{*}\cdot(\Lambda-\Lambda^{-1})=(\Lambda-\Lambda^{-1})\cdot B,
\end{align*}
then
\begin{align}
\left(\Delta\cdot A \cdot\iota_{\Lambda^{-1}}\Delta^{-1}\right)_{[0]}+\left(\Delta\cdot B \cdot\iota_{\Lambda}\Delta^{-1}\right)_{[0]}=
2\left(((\Delta\cdot A \cdot\iota_{\Lambda^{-1}}\Delta^{-1})_{\geq0})^*-((\Delta \cdot B \cdot \iota_{\Lambda}\Delta^{-1})_{<0})^*\right)\big{|}_{\Lambda=1}.\label{operatorarelation2}
\end{align}
\end{lemma}
\begin{proof}
Firstly similar to Lemma \ref{operatorequvalent}, if set $A=\tilde{A}(\Lambda-\Lambda^{-1})$ and $B=\tilde{B}(\Lambda^{-1}-\Lambda)$, then  $\tilde{A}^*=\tilde{B}$.  So if denote $\tilde{A}=\sum_{i}\tilde{A}_i(n)\Lambda^i$,  we can obtain
\begin{align*}
&\left(\Delta\cdot A \cdot\iota_{\Lambda^{-1}}\Delta^{-1}\right)_{[0]}+
\left(\Delta \cdot B \cdot\iota_{\Lambda}\Delta^{-1}\right)_{[0]}
=2\left(\tilde{A}_{-1}(n+1)
-\tilde{A}_{1}(n)\right).
\end{align*}
While
\begin{align*}
&\left((\Delta \cdot A \cdot \iota_{\Lambda^{-1}}\Delta^{-1})_{\geq0}\right)^*\big{|}_{\Lambda=1}
=\tilde{A}_{-1}(n+1)+\tilde{A}_0(n+1),\quad
&\left((\Delta\cdot B \cdot\iota_{\Lambda^{-1}}\Delta^{-1})_{<0}\right)^*\big{|}_{\Lambda=1}
=\tilde{A}_{1}(n)+\tilde{A}_0(n+1).
\end{align*}
Therefore we can get \eqref{operatorarelation2}.
\end{proof}
Similarly  by \eqref{todafirst}, \eqref{adjointeg} and Lemma \ref{operatorequvalentDelta}, we can get the following proposition.
\begin{proposition}\label{balancecondition2}
The balance condition II \eqref{condition2} is invariant under the flows $\pa_{\mathbf{t}_k}$,
$$\pa_{\mathbf{t}_k}\Big(e^{\varphi(n)}-b_0(n+1)^{-2}\Big)=0.$$
\end{proposition}
\begin{remark}
Given Toda Lax operators  $(\mathcal{L}_1,\mathcal{L}_2)$, if let
  \begin{align}
 L^+=\iota_{\Lambda^{-1}}\Delta^{-1}\cdot r\cdot\mathcal{L}_1 \cdot r^{-1}\cdot\Delta,\quad
L^-=\iota_{\Lambda}\Delta^{-1}\cdot r\cdot\mathcal{L}_2 \cdot r^{-1}\cdot \Delta,\label{largeBKPlaxop2}
  \end{align}
  satisfy large BKP constraints $(L^{+})^*{\left(\Lambda-\Lambda^{-1}\right)}=(\Lambda-\Lambda^{-1}) L^{-},$
 then  we can get $\pa_{\mathbf{t}_k}L^{\pm}=[B_{k},L^{\pm}],$
   where
 $$B_k=\iota_{\Lambda^{-1}}\Delta^{-1}\cdot r \cdot(\mathcal{L}_1^k)_{\geq0}\cdot r^{-1}\cdot\Delta
 -\iota_{\Lambda}\Delta^{-1}\cdot r \cdot(\mathcal{L}_2^k)_{<0}\cdot r^{-1}\cdot \Delta
 +\iota_{\Lambda^{-1}}\Delta^{-1}\cdot \pa_{{x}_k}r\cdot r^{-1} \cdot\Delta-\iota_{\Lambda}\Delta^{-1}\cdot \pa_{{y}_k}r \cdot r^{-1}\cdot \Delta.$$
Similarly by $u^+_{-1}=u^-_{-1}$, we can get balance condition II:
$e^{\varphi(n)}=r(n)^{-2}$. Next by Lemma \ref{operatorequvalentDelta} and $\pa_{\mathbf{t}_k}\varphi(n)=(\mathcal{L}_1^k)_{[0]}+(\mathcal{L}_2^k)_{[0]}$ {\rm(see \eqref{todafirst})},
 we can find
 $\pa_{\mathbf{t}_k}r=\left(-((\mathcal{L}_1)^k_{\geq0})^*
+((\mathcal{L}_2)^k_{<0})^*\right)(r)$.
Finally by using for operator $A$ and function $f$ \cite{Cheng12019}
 \begin{align*}
 &(\iota_{\Lambda^{-1}}\Delta^{-1} \cdot f \cdot A \cdot f^{-1}\cdot\Delta)_{\Delta,\geq 1}=\iota_{\Lambda^{-1}}\Delta^{-1} \cdot f \cdot A_{\geq 0} \cdot f^{-1} \cdot \Delta-\iota_{\Lambda^{-1}} \cdot \Delta^{-1} \cdot A^*_{\geq 0}(f) \cdot f^{-1} \cdot \Delta,\\
&(\iota_{\Lambda} \Delta^{-1} \cdot f \cdot A \cdot f^{-1} \cdot \Delta)_{\Delta^*,\geq 1}=\iota_{\Lambda}\Delta^{-1} \cdot f \cdot A_{< 0} \cdot f^{-1} \cdot \Delta-\iota_{\Lambda}\Delta^{-1} \cdot A^*_{< 0}(f) \cdot f^{-1} \cdot \Delta,
 \end{align*}
 we can derive $B_k=(L^+)^k_{\Delta,\geq1}-(L^-)^k_{\Delta^*,\geq1}.$
 Therefore this transformation \eqref{largeBKPlaxop2} is just anti--Miura transformation $\mathcal{T}_2=\Delta^{-1}r$.
 \end{remark}
By \eqref{largeBKPandmToda}, \eqref{largeBKPandmToda2} and Proposition \ref{muratras}, we can find
\begin{align*}
&\Psi_1(n,\mathbf{t},-\mathbf{t},z)=z^{-1}b_0^{-1}(n+1)
\Big(\Psi^+(n+1,\mathbf{t},z)-\Psi^+(n,\mathbf{t},z)\Big),\\
&\Psi_2(n,\mathbf{t},-\mathbf{t},z)=-z^{-1}b_0^{-1}(n+1)
\Big({\Psi^-}(n+1,\mathbf{t},z^{-1})-{\Psi^-}(n,\mathbf{t},z^{-1})\Big),\\
&\Psi_1^*(n,\mathbf{t},-\mathbf{t},z)=
zb_0(n+1)\Big(\Psi^-(n,\mathbf{t},z)+\Psi^-(n+1,\mathbf{t},z)\Big),\\
&\Psi_2^*(n,\mathbf{t},-\mathbf{t},z)=z^{-1}b_0(n+1)
\Big(\Psi^+(n,\mathbf{t},z^{-1})+\Psi^+(n+1,\mathbf{t},z^{-1})\Big),
\end{align*}
 satisfy Toda bilinear equation \eqref{TodaMiura}.

Let us summarize above results in the following proposition.
\begin{proposition}\label{Pro:miura2}
Given the large BKP objects: Lax operators $L^{\pm}$, wave operators ${S}^{\pm}$, wave functions $\Psi^{\pm}$, if set $\mathcal{L}_i$, $\mathcal{S}_i$, $\Psi_i$ and $\Psi_i^*$ in the form of Table III
\begin{center}
\begin{tabular}{llll}
\multicolumn{4}{c}{Table III. Miura transformations: large BKP hierarchy $\rightarrow$ Toda hierarchy}\\
\hline \hline
 & &$T_1=b_0(n)^{-1}$ &$T_2=b_0(n+1)^{-1}\Delta$  \\
\hline
{\rm Wave operator}&$\mathcal{S}_1=$ &${T}_1\cdot{S}^+$& ${T}_2\cdot{S}^+\Lambda^{-1}$\\
&$\mathcal{S}_2=$&${T}_1\cdot({S}^-)^*$& $-{T}_2\cdot({S}^-)^*$\\
{\rm Lax operator} &$\mathcal{L}_1=$ &
${T}_1\cdot{L}^+\cdot{T}_1^{-1}$& $T_2\cdot L_1\cdot\iota_{\Lambda^{-1}}T_2^{-1}$
 \\
 &$\mathcal{L}_2=$&$
{T}_1\cdot{L}^-\cdot{T}_1^{-1}$ & $ T_2\cdot L_2\cdot\iota_{\Lambda}T_2^{-1}$
 \\
 {\rm Wave function}&$\Psi_1(n,\mathbf{t},-\mathbf{t},z)=$&${T}_1\Psi^+(n,\mathbf{t},z)$&
$z^{-1}T_2\Big(\Psi^+(n,\mathbf{t},z)\Big)$\\
 &$\Psi_2(n,\mathbf{t},-\mathbf{t},z)=$&$z^{-1}{T}_1\Psi^-(n,\mathbf{t},z^{-1})$
 &$-z^{-1}T_2\Big({\Psi^-}(n,\mathbf{t},z^{-1})\Big)$\\
 {\rm Adjoint-- }&
$\Psi_1(n,\mathbf{t},-\mathbf{t},z)^*$=&${T}_1^{-1}
(\Lambda^{-1}-\Lambda)\Big(\Psi^-(n,\mathbf{t},z)\Big)$
&
$zb_0(n+1)(\Lambda+1)\Big(\Psi^-(n,\mathbf{t},z)\Big)$\\
{\rm wave function}&$\Psi_2(n,\mathbf{t},-\mathbf{t},z)^*=$&$
-z^{-1}{T}_1^{-1}(\Lambda^{-1}-\Lambda)\Big(\Psi^+(n,\mathbf{t},z^{-1})\Big)$
&$z^{-1}b_0(n+1)
(\Lambda+1)\Big(\Psi^+(n,\mathbf{t},z^{-1})\Big)$\\
\hline
\end{tabular}
\end{center}
\noindent then $\Psi_i$ and $\Psi_i^*$ satisfy \eqref{TodaMiura}, and $(\mathcal{L}_i, \mathcal{S}_i)$ satisfies Toda  evolution equations  \eqref{muiralax1}, \eqref{muiralax2} and \eqref{Todalaxmiura2}.
\end{proposition}
Next  the relations of tau functions for   large BKP and  Toda  hierarchies are given in the following theorem.
\begin{theorem}\label{taurelation}
Assume   $\tau^{{\rm Toda}}_n(\mathbf{t},-\mathbf{t})$ to be the Toda tau function derived from large BKP by Miura transformations $T_i$ $(i=1,2)$, then
\begin{align*}
\tau^{{\rm Toda}}_n(\mathbf{t},-\mathbf{t})=\left\{\begin{matrix}
{\rm const}\cdot \tau_n(\mathbf{t})\tau_{n-1}(\mathbf{t}),  &  {\rm in}\quad T_1,\\
{\rm const}\cdot  \tau^2_n(\mathbf{t}),& {\rm in}\quad T_2,
\end{matrix}\right.
\end{align*}
where $\tau_n(\mathbf{t})$ is large BKP  tau function.
\end{theorem}
\begin{proof}
Firstly by   Theorem  \ref{Theorem:taufuncyion}, we can get
$b_0(n)=\frac{\tau_{n-1}(\mathbf{t})
}{\tau_{n}(\mathbf{t})}$.
Then for Miura transformation $T_1$,   we can know by  Proposition \ref{Pro:miura2},
\begin{align*}
&\Psi_1(n,\mathbf{t},-\mathbf{t},z)=b_0(n)^{-1}\Psi^+(n,\mathbf{t},z),
\quad\Psi_2(n,\mathbf{t},-\mathbf{t},z)=z^{-1}b_0(n)^{-1}\Psi^-(n,\mathbf{t},z^{-1}).
\end{align*}
By \eqref{TLtau}, it can be found that
\begin{align}
&\frac{\tau^{{\rm Toda}}_n(\mathbf{t}-[z^{-1}],-{\mathbf{t}})}
{\tau^{{\rm Toda}}_n(\mathbf{t},-{\mathbf{t}})}
=\frac{\tau_{n-1}(\mathbf{t}-[z^{-1}])}{\tau_{n-1}(\mathbf{t})},\label{tau1relation}\\
&\frac{\tau^{{\rm Toda}}_{n+1}(\mathbf{t},-{\mathbf{t}}-[z])}
{\tau^{{\rm Toda}}_n(\mathbf{t},-{\mathbf{t}})}
=\frac{\tau_{n+1}(\mathbf{t}+[z])}{\tau_{n-1}(\mathbf{t})}.\label{tau2relation}
\end{align}

Next if let $f(n,x,y)={\rm log}\tau^{{\rm Toda}}_n(x,y),$
$g(n,x)={\rm log}\tau_n(x),$
then \eqref{tau1relation} and \eqref{tau2relation}  become
\begin{align}
&g(n-1,\mathbf{t}-[z^{-1}])-g(n-1,\mathbf{t})
=f(n,\mathbf{t}-[z^{-1}],-\mathbf{t})-f(n,\mathbf{t},-\mathbf{t}),\label{tau3relation}\\
&g(n+1,\mathbf{t}+[z])-g(n-1,\mathbf{t})
=f(n+1,\mathbf{t},-\mathbf{t}-[z])-f(n,\mathbf{t},-\mathbf{t}).\label{tau4relation}
\end{align}
Further if  set  $n-1\rightarrow n$, $\mathbf{t}-[z^{-1}]\rightarrow\mathbf{t}$ and $z^{-1}\rightarrow z$ in \eqref{tau3relation}, then we can get
\begin{align*}
&g(n,\mathbf{t})-g(n,\mathbf{t}+[z])
=f(n+1,\mathbf{t},-\mathbf{t}-[z])-f(n+1,\mathbf{t}+[z],-\mathbf{t}-[z]).
\end{align*}

If subtract the above equation from \eqref{tau4relation}, we can  get
\begin{align*}
&\Big(g(n+1,\mathbf{t}+[z])+g(n,\mathbf{t}+[z])\Big)
-\Big(g(n,\mathbf{t})+g(n-1,\mathbf{t})\Big)
=f(n+1,\mathbf{t}+[z],-\mathbf{t}-[z])-f(n,\mathbf{t},-\mathbf{t}).
\end{align*}
Further by the fact that  if function $h(n,\mathbf{t})$ satisfy $(\Lambda e^{\xi(\tilde{\pa},z)}-1)h(n,\mathbf{t})=0$, then $h(n,\mathbf{t})$ is a constant, therefore we can obtain
$$
g(n,\mathbf{t})+g(n-1,\mathbf{t})=f(n,\mathbf{t},-\mathbf{t})+{\rm const},$$
which implies
\begin{align*}
\tau^{{\rm Toda}}_{n}(\mathbf{t},-\mathbf{t})={\rm const}\cdot\tau_n(\mathbf{t})\tau_{n-1}(\mathbf{t}).
\end{align*}

While for Miura transformations $T_2$,  we can get by Proposition \ref{Pro:miura2}
\begin{align*}
&\Psi_1(n,\mathbf{t},-\mathbf{t},z)=z^{-1}b_0^{-1}(n+1)
\Big(\Psi^+(n+1,\mathbf{t},z)-\Psi^+(n,\mathbf{t},z)\Big),\\
&\Psi_2(n,\mathbf{t},-\mathbf{t},z)=-z^{-1}b_0^{-1}(n+1)
\Big({\Psi^-}(n+1,\mathbf{t},z^{-1})-{\Psi^-}(n,\mathbf{t},z^{-1})\Big).
\end{align*}
By \eqref{TLtau}, we can get
\begin{align}
&\frac{\tau^{{\rm Toda}}_n(\mathbf{t}-[z^{-1}],-{\mathbf{t}})}
{\tau^{{\rm Toda}}_n(\mathbf{t},-{\mathbf{t}})}
=\frac{\tau_{n}(\mathbf{t}-[z^{-1}])\tau_{n}(\mathbf{t})}{\tau_{n}^2(\mathbf{t})}
-\frac{\tau_{n+1}(\mathbf{t})\tau_{n-1}(\mathbf{t}-[z^{-1}])}{\tau_{n}^2(\mathbf{t})}z^{-1},\label{2tau1relation}\\
&\frac{\tau^{{\rm Toda}}_{n+1}(\mathbf{t},-{\mathbf{t}}-[z])}
{\tau^{{\rm Toda}}_n(\mathbf{t},-{\mathbf{t}})}
=\frac{\tau_{n+1}(\mathbf{t}+[z])\tau_{n+1}(\mathbf{t})}{\tau_{n}^2(\mathbf{t})}
-\frac{\tau_{n+2}(\mathbf{t}+[z])\tau_{n}(\mathbf{t})}{\tau_{n}^2(\mathbf{t})}z.\label{2tau2relation}
\end{align}

Further if  set  $n-1\rightarrow n$, $\mathbf{t}-[z^{-1}]\rightarrow\mathbf{t}$ and $z^{-1}\rightarrow z$  in \eqref{2tau1relation},
\begin{align*}
\frac{\tau^{{\rm Toda}}_{n+1}(\mathbf{t},-{\mathbf{t}}-[z])}
{\tau^{{\rm Toda}}_{n+1}(\mathbf{t}+[z],-{\mathbf{t}}-[z])}
=\frac{\tau_{n+1}(\mathbf{t}+[z])\tau_{n+1}(\mathbf{t})}{\tau_{n+1}^2(\mathbf{t}+[z])}
-\frac{\tau_{n+2}(\mathbf{t}+[z])\tau_{n}(\mathbf{t})}{\tau_{n+1}^2(\mathbf{t}+[z])}z.\label{2tau3relation}
\end{align*}
Next according to  \eqref{2tau2relation},
\begin{align*}
\frac{\tau^{{\rm Toda}}_{n}(\mathbf{t},-{\mathbf{t}})}{\tau^{{\rm Toda}}_{n+1}(\mathbf{t}+[z],-{\mathbf{t}}-[z])}=\frac{\tau_n^2(\mathbf{t})}
{\tau_{n+1}^2(\mathbf{t}+[z])},
\end{align*}
which implies $$
f(n,\mathbf{t},-\mathbf{t})-f(n+1,\mathbf{t}+[z],-\mathbf{t}-[z])=2 \Big(g(n,\mathbf{t})-g(n+1,\mathbf{t}+[z])\Big).$$
Therefore
$$f(n,\mathbf{t},-\mathbf{t})=2g(n,\mathbf{t})+{\rm const},$$
which is just
$\tau^{{\rm Toda}}_{n}(\mathbf{t},-{\mathbf{t}})={\rm const}\cdot\tau_n^2(\mathbf{t}).$
\end{proof}
Further according to Remark \ref{taufunctionmToda}, we have the corollary below.
\begin{corollary}\label{mTodatauToda}
Given large BKP tau function  $\tau_n(\mathbf{t})$, we can find $\tau_n(\mathbf{t})\tau_{n-1}(\mathbf{t})$ and $\tau_n^2(\mathbf{t})$ are different Toda tau functions, if setting ${\tau}^{\rm Toda}_{n}(\mathbf{t},-\mathbf{t})=\tau_n(\mathbf{t})\tau_{n-1}(\mathbf{t})$ and $\tilde{\tau}_{n}^{\rm Toda}(\mathbf{t},-\mathbf{t})=\tau_n^2(\mathbf{t})$, then $\big(\tau_{n}^{\rm Toda}(\mathbf{t},-\mathbf{t}),\tilde{\tau}_{n}^{\rm Toda}(\mathbf{t},-\mathbf{t})\big)$ satisfies bilinear equation of mToda hierarchy, that is
\begin{align*}
&{\rm Res}_{z}\left(z^{n-n'-1}\tau_{n}^{\rm Toda}(\mathbf{t}-[z^{-1}],-\mathbf{t})
\cdot\tilde{\tau}_{n'}^{\rm Toda}(\mathbf{t}'+{[z^{-1}]},-\mathbf{t})e^{\xi(\mathbf{t}-\mathbf{t}',z)}\right.\nonumber\\
&\quad\quad\quad\left.
+z^{n'-n-2}\tau_{n+1}^{\rm Toda}(\mathbf{t},-\mathbf{t}-{[z^{-1}]})
\cdot\tilde{\tau}_{n'-1}^{\rm Toda}(\mathbf{t}',-\mathbf{t}'+[z^{-1}])e^{-\xi(\mathbf{t}-\mathbf{t}',z)}\right)=\tilde{\tau}_{n}^{\rm Toda}
(\mathbf{t},-\mathbf{t})\tau_{n'}^{\rm Toda}(\mathbf{t}',-\mathbf{t}').
\end{align*}

\end{corollary}
\begin{remark}
The results in Theorem \ref{taurelation} and Corollary \ref{mTodatauToda} should be the generalization of relations of tau functions for modified BKP and modified KP hierarchies (see Theorem 24 in \cite{Guan2022}).
\end{remark}

\section{Conclusion and discussion}\label{section6}
In this paper, the Lax structure and tau function for the large BKP hierarchy are discussed. Firstly from bilinear equation, we obtain the corresponding Lax equation as follows,
\begin{align*}
&\pa_{\mathbf{t}_k}L^{\pm}=[B_{k},L^{\pm}],\quad (L^+)^*=(\Lambda-\Lambda^{-1})\cdot L^{-}\cdot\iota_{\Lambda}{\left(\Lambda-\Lambda^{-1}\right)^{-1}},\\
&B_{k}=(L^{+})^k_{\geq 1}-(L^{-})^k_{\leq -1}-\left((L^{+})^k_{\geq 1}-(L^{-})^k_{\leq -1}\right)|_{\Lambda=1}=(L^+)^k_{\Delta,\geq1}-(L^-)^k_{\Delta^*,\geq1},
\end{align*}
where in particular the relation of flow generator $B_k$ with Lax operators $L^{\pm}$ is obtained.  Also starting from Lax equation, the corresponding bilinear equation and existence of tau function are discussed, that is
 \begin{align*}
{\rm Res}_{z}{z^{-1}}\left(\Psi^+(n',\mathbf{t'},z)
\Psi^-(n'',\mathbf{t}'',z)+\Psi^-(n',\mathbf{t'},z)
\Psi^+(n'',\mathbf{t}'',z)\right)=\frac{1}{2}\left(1-(-1)^{n'-n''}\right),
\end{align*}
and
\begin{align*}
\Psi^{+}(n,\mathbf{t},z)
=\frac{\tau_{n-1}(\mathbf{t}-[z^{-1}])}{\tau_n(\mathbf{t})}
e^{\xi(\mathbf{t},z)}z^{n},\quad
\Psi^{-}(n,\mathbf{t},z)
=\frac{\tau_{n+1}(\mathbf{t}+[z^{-1}])}{\tau_n(\mathbf{t})}
e^{-\xi(\mathbf{t},z)}z^{-n-1}.
\end{align*}
After that, large BKP hierarchy is viewed as sub--hierarchy of mToda hierarchy. And based upon this, we understand two typical relations of tau functions between large BKP and Toda hierarchies by using two basic Miura transformations from mToda to Toda, that is,
\begin{align*}
\tau^{{\rm Toda}}_n(\mathbf{t},-\mathbf{t})=\left\{\begin{matrix}
{\rm const}\cdot \tau_n(\mathbf{t})\tau_{n-1}(\mathbf{t}),  &  {\rm in}\quad T_1=b_0(n)^{-1},\\
{\rm const}\cdot  \tau^2_n(\mathbf{t}),& {\rm in}\quad T_2=b_0(n+1)^{-1}\Delta.
\end{matrix}\right.
\end{align*}
And $(\tau_n(\mathbf{t})\tau_{n-1}(\mathbf{t}),\tau_n^2(\mathbf{t}))$ satisfies bilinear equation of the mToda hierarchy.

In terms of Lax structure, we can do many works for large BKP hierarchy, such as Darboux transformations, reductions, additional symmetries and Hamiltonian structures. Besides the basic relations above for tau functions between large BKP and Toda, we believe there are also other relations. Notice that ${\rm Ad}\Lambda^i$ is the
automorphism of Toda and mToda hierarchies, then by the following diagram\\

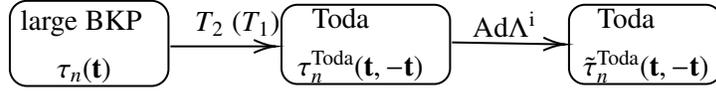
\begin{center}

\tikzset{every picture/.style={line width=0.75pt}} 

\begin{tikzpicture}[x=0.75pt,y=0.75pt,yscale=-1,xscale=1]

\draw   (101,120.6) .. controls (101,115.85) and (104.85,112) .. (109.6,112) -- (172.4,112) .. controls (177.15,112) and (181,115.85) .. (181,120.6) -- (181,146.4) .. controls (181,151.15) and (177.15,155) .. (172.4,155) -- (109.6,155) .. controls (104.85,155) and (101,151.15) .. (101,146.4) -- cycle ;
\draw    (182,135) -- (233,135) ;
\draw [shift={(235,135)}, rotate = 180] [color={rgb, 255:red, 0; green, 0; blue, 0 }  ][line width=0.75]    (10.93,-3.29) .. controls (6.95,-1.4) and (3.31,-0.3) .. (0,0) .. controls (3.31,0.3) and (6.95,1.4) .. (10.93,3.29)   ;
\draw   (238,122) .. controls (238,117.58) and (241.58,114) .. (246,114) -- (315.03,114) .. controls (319.45,114) and (323.03,117.58) .. (323.03,122) -- (323.03,146) .. controls (323.03,150.42) and (319.45,154) .. (315.03,154) -- (246,154) .. controls (241.58,154) and (238,150.42) .. (238,146) -- cycle ;
\draw    (324,133) -- (339.03,133.26) -- (379,133.96) ;
\draw [shift={(381,134)}, rotate = 181.01] [color={rgb, 255:red, 0; green, 0; blue, 0 }  ][line width=0.75]    (10.93,-3.29) .. controls (6.95,-1.4) and (3.31,-0.3) .. (0,0) .. controls (3.31,0.3) and (6.95,1.4) .. (10.93,3.29)   ;
\draw   (383,122) .. controls (383,117.58) and (386.58,114) .. (391,114) -- (452,114) .. controls (456.42,114) and (460,117.58) .. (460,122) -- (460,146) .. controls (460,150.42) and (456.42,154) .. (452,154) -- (391,154) .. controls (386.58,154) and (383,150.42) .. (383,146) -- cycle ;

\draw (105,115) node [anchor=north west][inner sep=0.75pt]   [align=left] {large BKP\\ \ \ \ \ \ $\tau_n(\mathbf{t})$\\ \ \ \ \ \ };
\draw (193,116) node [anchor=north west][inner sep=0.75pt]   [align=left] {$T_2$ $(T_1)$};
\draw (333,116) node [anchor=north west][inner sep=0.75pt]   [align=left] {${\rm Ad \Lambda^i}$};
\draw (241,114) node [anchor=north west][inner sep=0.75pt]   [align=left] { \ \ \ Toda \ \ \\ \ $\tau_n^{\rm Toda}(\mathbf{t},-\mathbf{t})$};
\draw (385,114) node [anchor=north west][inner sep=0.75pt]   [align=left] { \ \ \ Toda \ \ \\ \ $\tilde{\tau}_n^{\rm Toda}(\mathbf{t},-\mathbf{t})$};

\end{tikzpicture}
\end{center}
we can find
\begin{align*}
\tilde{\tau}_n^{\rm Toda}(\mathbf{t},-\mathbf{t})=\tau_{n+i}^{{\rm Toda}}(\mathbf{t},-\mathbf{t})=\left\{\begin{matrix}
{\rm const}\cdot \tau_{n+i}(\mathbf{t})\tau_{n+i-1}(\mathbf{t}),  &  {\rm in}\quad T_1=b_0(n)^{-1},\\
{\rm const}\cdot\tau_{n+i}^2(\mathbf{t}),& {\rm in}\quad T_2=b_0(n+1)^{-1}\Delta.
\end{matrix}\right.
\end{align*}
Similarly according to the diagram below,\\
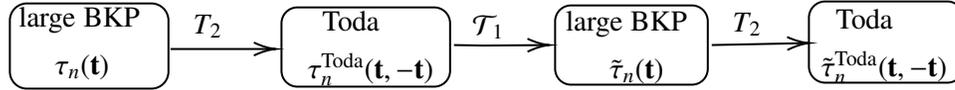
\begin{center}
\tikzset{every picture/.style={line width=0.75pt}} 

\begin{tikzpicture}[x=0.75pt,y=0.75pt,yscale=-1,xscale=1]

\draw   (101,120.6) .. controls (101,115.85) and (104.85,112) .. (109.6,112) -- (172.4,112) .. controls (177.15,112) and (181,115.85) .. (181,120.6) -- (181,146.4) .. controls (181,151.15) and (177.15,155) .. (172.4,155) -- (109.6,155) .. controls (104.85,155) and (101,151.15) .. (101,146.4) -- cycle ;
\draw    (182,135) -- (233,135) ;
\draw [shift={(235,135)}, rotate = 180] [color={rgb, 255:red, 0; green, 0; blue, 0 }  ][line width=0.75]    (10.93,-3.29) .. controls (6.95,-1.4) and (3.31,-0.3) .. (0,0) .. controls (3.31,0.3) and (6.95,1.4) .. (10.93,3.29)   ;
\draw   (238,122) .. controls (238,117.58) and (241.58,114) .. (246,114) -- (315.03,114) .. controls (319.45,114) and (323.03,117.58) .. (323.03,122) -- (323.03,146) .. controls (323.03,150.42) and (319.45,154) .. (315.03,154) -- (246,154) .. controls (241.58,154) and (238,150.42) .. (238,146) -- cycle ;
\draw    (325,133) -- (340.03,133.26) -- (371,133.02) ;
\draw [shift={(373,133)}, rotate = 179.54] [color={rgb, 255:red, 0; green, 0; blue, 0 }  ][line width=0.75]    (10.93,-3.29) .. controls (6.95,-1.4) and (3.31,-0.3) .. (0,0) .. controls (3.31,0.3) and (6.95,1.4) .. (10.93,3.29)   ;
\draw   (376,121) .. controls (376,116.58) and (379.58,113) .. (384,113) -- (445,113) .. controls (449.42,113) and (453,116.58) .. (453,121) -- (453,145) .. controls (453,149.42) and (449.42,153) .. (445,153) -- (384,153) .. controls (379.58,153) and (376,149.42) .. (376,145) -- cycle ;
\draw   (504,119.2) .. controls (504,114.67) and (507.67,111) .. (512.2,111) -- (573.8,111) .. controls (578.33,111) and (582,114.67) .. (582,119.2) -- (582,143.8) .. controls (582,148.33) and (578.33,152) .. (573.8,152) -- (512.2,152) .. controls (507.67,152) and (504,148.33) .. (504,143.8) -- cycle ;
\draw    (455,133) -- (502,132.04) ;
\draw [shift={(504,132)}, rotate = 178.83] [color={rgb, 255:red, 0; green, 0; blue, 0 }  ][line width=0.75]    (10.93,-3.29) .. controls (6.95,-1.4) and (3.31,-0.3) .. (0,0) .. controls (3.31,0.3) and (6.95,1.4) .. (10.93,3.29)   ;

\draw (104,112) node [anchor=north west][inner sep=0.75pt]   [align=left] {large BKP\\ \ \ \ \ \ $\tau_n({\mathbf{t}})$\\ \ \ \ \ \ };
\draw (192,116) node [anchor=north west][inner sep=0.75pt]   [align=left] {$T_2$};
\draw (333,116) node [anchor=north west][inner sep=0.75pt]   [align=left] {$\mathcal{T}_1$};
\draw (246,114) node [anchor=north west][inner sep=0.75pt]   [align=left] { \ \ \ Toda \ \ \\ \ $\tau_n^{\rm Toda}(\mathbf{t},-\mathbf{t})$};
\draw (380,114) node [anchor=north west][inner sep=0.75pt]   [align=left] {large BKP\\ \ \ \ \ \ \ $\tilde{\tau}_n({\mathbf{t}})$};
\draw (464,115) node [anchor=north west][inner sep=0.75pt]   [align=left] {$T_2$};
\draw (505,113) node [anchor=north west][inner sep=0.75pt]   [align=left] { \ \ \ Toda \ \ \\ \ $\tilde{\tau}_n^{\rm Toda}(\mathbf{t},-\mathbf{t})$};

\end{tikzpicture}
\end{center}
we can know
$$\tilde{\tau}_{n}^{{\rm Toda}}(\mathbf{t},-\mathbf{t})\tilde{\tau}_{n-1}^{{\rm Toda}}(\mathbf{t},-\mathbf{t})={\rm const}\cdot\tau_{n}^4(\mathbf{t}).$$\\
\noindent{\bf Acknowledgements}: { Many thanks to Dr. Prokofev for useful information and comments. This work is supported by National Natural Science Foundations of China (Nos. 12171472 and 12261072) and ``Qinglan Project" of Jiangsu Universities. }\\

\noindent{\bf Conflict of Interest}:
The author have no conflicts to disclose.\\

\noindent{\bf Data availability}: Date sharing is not applicable to this article as no new data were created or analyzed in this study.\\

\end{document}